\documentclass[runningheads]{llncs}
\newcommand{\version}{0}
\newcommand{\SLV}[2]{\ifthenelse{\equal{\version}{0}}{#1}{ \RED{#2}}}

\usepackage[T1]{fontenc}
\usepackage{graphicx}
\usepackage{subcaption} 
\usepackage{hyperref}
\usepackage{color}

\AtBeginDocument{%
  }

\usepackage{tikz} 
\usetikzlibrary {graphs} 
\usepackage{wrapfig} 

\usepackage{xspace}
\def\eg{{\it e.g.}\xspace}
\def\ie{{\it i.e.}\xspace}
\def\etc{{\it etc.}\xspace}

\newcommand\todom[2][]{\todo[color=blue!20,#1]{#2}} 
\newcommand\todoc[2][]{\todo[color=pink!20,#1]{#2}} 

\newcommand{\RED}[1]{{\color{red}{#1}}}
\newcommand{\blue}[1]{{\color{blue}{#1}}}

\newcommand{\textdef}[1]{\emph{#1}} 
\newcommand{\st}{\mathrel{|}}

\newcommand{\Def}{\mathrel{:=}} 
\newcommand{\Bnf}{\mathrel{::=}} 
\newcommand{\St}{\mathrel{|}}

	\newcommand{\Bool}{\ensuremath{\mathsf{Bool}}} 
	\newcommand{\PosA}{P} 
	\newcommand{\PosB}{Q} 

	\newcommand{\NegA}{A} 
	\newcommand{\NegB}{B} 

	\newcommand{\TypeA}{T} 
	\newcommand{\TypeB}{S} 
	
	\newcommand{\Type}[1]{\mathsf{ty}(#1)} 

\newcommand{\Coin}[1]{\mathtt{coin}_{#1}}  
	\newcommand{\LetCalculus}{\mathcal L} 
	\newcommand{\FV}[2][]{\mathsf{FV}(#2)^{#1}} 					
	\newcommand{\App}[2]{#1 #2} 								
	\newcommand{\Abs}[2]{\lambda #1.#2} 						
	\newcommand{\Seq}[1]{(#1)} 								
	\newcommand{\Let}[3]{\mathtt{let}\, #1\, =\, #2\, \mathtt{in}\, #3} 	
	\newcommand{\LetG}[2]{(#1\, \mathtt{in}\, #2)} 					
	\newcommand{\StocA}{\ensuremath{\mathtt M}} 				

	\newcommand{\LetTerm}{\ensuremath{\ell}} 				
	\newcommand{\SD}{\ensuremath{\mathsf{SD}}} 			
	\newcommand{\VA}{\ensuremath{\mathsf{VA}}} 			
	\newcommand{\VE}{\ensuremath{\mathsf{VE}}} 			

	\newcommand{\VEF}{\ensuremath{\mathsf{VE}^{\mathsf F}}} 			
	\newcommand{\VEL}{\ensuremath{\mathsf{VE}^\LetCalculus}} 			

	\newcommand{\Red}[1][]{\xrightarrow{#1}}	
	\newcommand{\Swap}{\gamma} 
	\newcommand{\Mult}{\mu} 
	\newcommand{\Elimin}{\epsilon} 

	\newcommand{\RP}{\ensuremath{\mathbb{R}_{\geq 0}}} 
	\newcommand{\Pcs}[1]{\ensuremath{\mathcal{#1}}} 
	\newcommand{\Clique}[1]{\ensuremath{\mathsf{P}(#1)}} 
	\newcommand{\Web}[1]{\vert #1 \vert} 
	\newcommand{\Sem}[2][]{\llbracket #2\rrbracket^{#1}} 

	\newcommand{\MApp}[2]{#1\cdot #2} 			

	\newcommand{\Facts}[1]{\mathsf{Fs}(#1)}		
	\newcommand{\Fact}[1]{\mathsf{F}(#1)}			
	\newcommand{\FComp}[2]{#1 :: #2}				
	\newcommand{\FProd}{\odot} 					
	\newcommand{\BigFProd}{\bigodot}				
	\newcommand{\FSum}[2]{\textstyle\sum_{#1}(#2)} 	

	\newcommand{\CollA}{\mathcal V} 
	\newcommand{\CollB}{\mathcal W} 

	\newcommand{\LabelA}{v} 

	\newcommand{\FactA}{\phi} 
	\newcommand{\FactB}{\psi} 

	\newcommand{\FactSetA}{\Gamma} 
	\newcommand{\FactSetB}{\Delta} 
	\newcommand{\FactSetC}{\Xi} 

	\newcommand{\Fam}{\mathsf{Fam}} 
	\newcommand{\FactFun}[1]{\mathsf{Fun}(#1)} 
	\newcommand{\FactVar}[2][]{\mathsf{Var}^{#1}(#2)} 

	\newcommand{\Deg}{\mathsf{d}} 
	\newcommand{\Base}{\mathsf{b}} 

	\newcommand{\Val}[1]{\Web{#1}} 
	\newcommand{\ValA}{a} 
	\newcommand{\ValB}{b} 
	\newcommand{\ValC}{c} 
	\newcommand{\FValA}{\overline a} 
	\newcommand{\FValB}{\overline b} 
	\newcommand{\FValC}{\overline c} 
	\newcommand{\FPlus}{\uplus} 
	\newcommand{\Proj}[3][]{#2\vert_{#3}^{#1}} 

	\newcommand{\SetA}{S} 
	\newcommand{\SetB}{{T}} 
	\newcommand{\SetC}{{U}} 

\newcommand{\Size}{\mathsf{s}}

\newcommand{\Nat}{\ensuremath{\mathbb{N}}} 

\newcommand{\ScalA}{\rho} 
\newcommand{\ScalB}{\tau} 

\newcommand{\VecA}{\alpha} 

\newcommand{\MatA}{\phi} 
\newcommand{\MatB}{\psi} 
\newcommand{\MatC}{\xi} 

\newcommand{\MProd}[2]{{#1}{#2}} 
\newcommand{\True}{\ensuremath{\mathtt{t}}} 
\newcommand{\False}{\ensuremath{\mathtt{f}}} 

\renewenvironment{itemize}
{
	\begin{list}{\labelitemi}
		{\setlength{\itemsep}{0pt}
			\setlength{\topsep}{0pt}
			\setlength{\parsep}{0pt}
			\setlength{\partopsep}{0pt}
			\setlength{\leftmargin}{15pt}
			\setlength{\rightmargin}{0pt}
			\setlength{\itemindent}{0pt}
			\setlength{\labelsep}{5pt}
			\setlength{\labelwidth}{10pt}
	}}
	{
	\end{list}
}
\newcounter{numberone}
\renewenvironment{enumerate}
{
	\begin{list}{\arabic{numberone}.}
		{
			\usecounter{numberone}
			\setlength{\itemsep}{0pt}
			\setlength{\topsep}{0pt}
			\setlength{\parsep}{0pt}
			\setlength{\partopsep}{0pt}
			\setlength{\leftmargin}{15pt}
			\setlength{\rightmargin}{0pt}
			\setlength{\itemindent}{0pt}
			\setlength{\labelsep}{5pt}
			\setlength{\labelwidth}{15pt}
	}}
	{
	\end{list}
}

\makeatletter
\renewcommand\paragraph{\@startsection{paragraph}{4}{\z@}%
	{-10\p@ \@plus -2\p@ \@minus -4\p@}%
	{-0.5em \@plus -0.22em \@minus -0.1em}%
	{\normalfont\normalsize\itshape}}
\makeatother

\usepackage[disable,
textwidth=1.5cm,textsize=footnotesize]{todonotes} 

\usepackage[normalem]{ulem}
\usepackage{amsmath} 	
\usepackage{stmaryrd} 
\usepackage{amssymb} 	
\usepackage{cmll} 		
\usepackage{proof} 		
\usepackage{bussproofs} 	
\usepackage{listings} 	
\usepackage{lscape}		
	\usepackage{afterpage}
\usepackage{mathtools}	
\usepackage{xcolor}		
\newcommand\Myqed{}

\begin{document}

\title{Variable Elimination as Rewriting\\ 
in a Linear Lambda Calculus
}

\author{
	Thomas Ehrhard\inst{1}\orcidID{0000-0001-5231-5504}
	\and
	Claudia Faggian\inst{1}
	\and
	Michele Pagani\inst{2}\orcidID{0000-0001-6271-3557}
}
%
%
\institute{
	Université de Paris  Cité, CNRS, IRIF, F-75013, Paris France
	\email{\{ehrhard,faggian\}@irif.fr}
	\and
	École Normale Supérieure, LIP, F-69342, Lyon, France
	\email{michele.pagani@ens-lyon.fr}
}

\maketitle

\begin{abstract}
	Variable Elimination ($\VE$) is a classical  \emph{exact inference} algorithm for  probabilistic graphical models such as Bayesian Networks, 
	computing the marginal distribution of a subset of the
	random variables in the model.
	Our goal is to understand Variable Elimination as an algorithm acting  \emph{on programs}, here expressed  in an idealized probabilistic functional language---a linear simply-typed $\lambda$-calculus suffices for our purpose.
	Precisely, we express $\VE$ as \emph{a term rewriting process},  
	which transforms a global definition of a variable into a local
	definition, by swapping and nesting let-in expressions. 
	We exploit in an essential way linear types.

%
\keywords{
	Linear Logic \and
	Lambda Calculus \and
	Bayesian Inference \and
	Probabilistic Programming \and
	Denotational Semantics
}
\end{abstract}



\section{Introduction}

Probabilistic programming  languages  (PPLs) provide a rich and expressive framework for 
stochastic modeling and Bayesian reasoning. 
The crucial but  computationally hard task  is that of  inference, \ie computing explicitly  the
probability distribution which is implicitly specified by the probabilistic program.
Most  PPLs focus on continuous random variables---in this setting the inference engine typically implements \emph{approximate} inference algorithms based on sampling methods (such as importance sampling, Markov Chain Monte Carlo,  Gibbs sampling).  
However, 
 several domains of application  (\eg network verification, ranking and voting, 
 text or graph analysis) are naturally discrete, yielding to an increasing interest in the challenge  of 
 \emph{exact inference}   \cite{HoltzenBM20,GorinovaGSV22,ZaiserMO23,ObermeyerBJPCRG19,ZhouYTR20,GehrMV16,NarayananCRSZ16,Pfeffer07}. 
  A good   example is Dice \cite{HoltzenBM20},  
  a \emph{first-order} functional language 
 whose inference algorithm  exploits \emph{the structure of the program} in order to  factorise inference, making it possible to scale exact inference to   large distributions. A common ground  to most exact approaches  is to be  inspired by techniques for exact inference on discrete graphical models, which typically  exploit probabilistic independence as the key for compact representation and efficient inference.

Indeed, specialized formalisms do come with highly efficient  algorithms  for   \emph{exact inference}; a prominent example is that of  Bayesian networks, which enable    algorithms such as  Message Passing \cite{Pearl88} and  Variable Elimination \cite{ZhangP94}---to name two classical ones---and a variety of approaches for  exploiting local structure, such as  reducing inference to  Weighted Model Counting \cite{ChaviraD05,ChaviraD08}.   General-purpose programming language do provide a rich expressiveness, which allows in particular for the encoding of  Bayesian networks, however, the corresponding  algorithms are often lost when leaving  the realm of graphical models for PPLs, leaving an uncomfortable gap between the two worlds. Our goal is shedding light in this gray area, understanding  exact  inference as an algorithm acting  \emph{on programs}.


In pioneering work,  Koller et al.~\cite{KollerMP97} define a general purpose  functional language which not only is able to encode Bayesian networks (as well as other specialized formalisms),  but also comes with  an algorithm which mimics Variable Elimination (\VE{} for short)  by means of \emph{term transformation}.  \VE{} is arguably the simplest  algorithm for exact inference, which  is 
factorised  into smaller intermediate computations,  by eliminating  the irrelevant variables according to a specific order.  
 The limit in \cite{KollerMP97}    is  that unfortunately,  the algorithm there can only implement a specific elimination ordering (the one determined  by the lazy evaluation  implicit in the algorithm), which   might not be  the  most efficient: a different ordering might result in smaller intermediate factors. 
The general problem  to be able to deal with any possible ordering, hence producing any possible factorisation, is there left as an open challenge for further investigation.  The approach that is taken by the authors in a series of  subsequent papers  will go  in a different direction from term rewriting; eventually in    \cite{Pfeffer07}  
programs are compiled into an intermediate structure, and it is on this graph structure that  a sophisticated variant of  \VE{} is performed.
The question of understanding  \VE{} as a \emph{transformation on programs} remains still open; we believe it is important for a  foundational understanding of PPLs.

 In this paper, we provide an answer, defining  an inference algorithm which \emph{fully}   formalizes the classical  \VE{}  algorithm  as rewriting of programs, expressed  in an idealized probabilistic functional language---a linear simply-typed $\lambda$-calculus suffices for our purpose.  Formally, we prove \emph{soundness and completeness} of our  algorithm with respect to the standard one.  Notice that the choice of the elimination order is not part of a \VE{} algorithm---several heuristics are available in the literature to compute an efficient elimination order (see \eg \cite{Darwiche2009}). As wanted, we prove that  \emph{any} given elimination ordering can be implemented by our algorithm. 
 When we run it on a stochastic program representing a Bayesian network,  its computational behaviour is the same as that of standard \VE{} for Bayesian networks, and the cost is of the same complexity order.
 While the idea behind \VE{} is simple, crafting  an algorithm  on terms which is able to 
 implement any elimination order is  non-trivial--- our success here relies on the use of linear types $\PosA\multimap\TypeA$ enabled by  linear logic \cite{ll}, accounting for the interdependences generated by a specific elimination order.
 Let us explain the main ideas.

  

%
%
%

\paragraph{Factorising Inference, via the graph structure.}
Bayesian networks describe a set of random variables and  their conditional (in)dependencies. Let us restrict ourselves to boolean random variables, \ie variables $x$ representing a boolean value $\True$ or $\False$ with some probability.\footnote{The results trivially extends to random variables over \emph{countable} sets of outcomes.} Such a variable can be described as a vector of two non-negative real numbers 
$(\ScalA_\True, \ScalA_\False)$ quantifying the probability $\ScalA_\True$ (resp.~$\ScalA_\False$) of sampling $\True$ (resp.~$\False$) from $x$.

\begin{figure}[t]
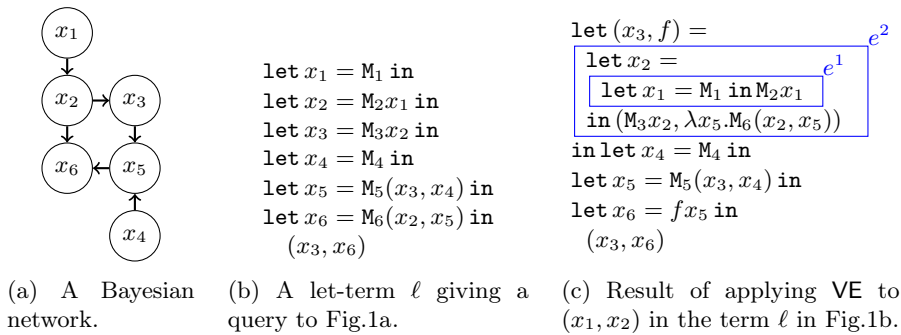

	\centering
	\begin{subfigure}[b]{2.5cm}
		\centering
		\tikz[scale=0.9]{
			\node [circle, draw] (1) at (0, 3) {$x_1$};	
			\node [circle, draw] (2) at (0, 2) {$x_2$};	
			\node [circle, draw] (3) at  (1, 2) {$x_3$};
			\node [circle, draw] (4) at (1, 0) {$x_4$};	
			\node [circle, draw] (5) at (1, 1) {$x_5$};	
			\node [circle, draw] (6) at (0, 1) {$x_6$};	
		
			\graph[edges={thick}]  { 
				(1) -> (2) -> (3) -> (5) -> (6) ;
				(2) -> (6) ;
				(4) -> (5) ;
			};
		}
		\caption{A Bayesian network.}\label{sub_fig:bayesian_graph}
	\end{subfigure}
	\quad
	\begin{subfigure}[b]{4cm}
		\centering
		\begin{minipage}{3.5cm}
			\[
				\begin{array}{l}
					\mathtt{let }\, x_1 = \StocA_1\, \mathtt{ in}\\
					\mathtt{let }\, x_2 = \StocA_2 x_1\, \mathtt{ in}\\
					\mathtt{let }\, x_3 = \StocA_3 x_2\, \mathtt{ in}\\
					\mathtt{let }\, x_4 = \StocA_4 \, \mathtt{ in}\\
					\mathtt{let }\, x_5 = \StocA_5 (x_3, x_4)\, \mathtt{ in}\\
					\mathtt{let }\, x_6 = \StocA_6 (x_2, x_5)\, \mathtt{ in}\\
				\quad(x_3, x_6)
				\end{array}
			\]
		\end{minipage}
		\caption{A let-term $\LetTerm$ giving a query to Fig.\ref{sub_fig:bayesian_graph}.}\label{sub_fig:let-term}
	\end{subfigure}
	\quad
	\begin{subfigure}[b]{4.5cm}
		\centering
		\tikz{
			\node (1) [align=left, text width=3.5cm] at (0, 2.8) {$\mathtt{let }\, (x_3, f) =$};	
			\node (2) [align=left, text width=3.5cm] at (0.2, 2.4) {$\mathtt{let }\, x_2 =$};	
			\node (3) [align=left, text width=3.5cm] at (0.4, 2) {
				$\mathtt{let }\, x_1 = \StocA_1 \, \mathtt{ in}\, \StocA_2 x_1$
			};	
			\node (4) [align=left, text width=3.5cm] at (0.2, 1.6) {
				$\mathtt{in }\, (\StocA_3 x_2, \lambda x_5. \StocA_6 (x_2, x_5))$
			};
			\node (5) [align=left, text width=3.5cm] at (0, 1.2) {
				$\mathtt{in }\,\mathtt{let }\, x_4 = \StocA_4 \, \mathtt{ in}$
			};
			\node (6) [align=left, text width=3.5cm] at (0, 0.8) {
				$\mathtt{let }\, x_5 = \StocA_5 (x_3, x_4)\, \mathtt{ in}$
			};
			\node (7) [align=left, text width=3.5cm] at (0, 0.4) {
				$\mathtt{let }\, x_6 = f x_5\, \mathtt{ in}$
			};
			\node (8) [align=left, text width=3.5cm] at (0.2, 0) {
				$(x_3, x_6)$
			};		
			\draw [draw=blue] (-1.5, 1.8) rectangle (1.6, 2.2);
			\draw [draw=blue] (-1.7, 1.4) rectangle (2.2, 2.6);
			\node (e1) [color=blue] at (1.75, 2.35) {$e^1$};
			\node (e2) [color=blue] at (2.35, 2.75) {$e^2$};
		}		
		\caption{Result of applying $\VE$ to $(x_1, x_2)$ in the term $\LetTerm$ in Fig.\ref{sub_fig:let-term}.}\label{sub_fig:let-term-VE}
	\end{subfigure}
\caption{Example of running the $\VE$ algorithm on a let-term $\LetTerm$.}\label{fig:intro_example}
\end{figure}

%
Fig.~\ref{sub_fig:bayesian_graph} depicts an example of Bayesian network. It is a directed acyclic graph $\mathcal G$ where the nodes are associated with random variables and where the arrows describe conditional dependencies between these variables. For instance, in Fig.~\ref{sub_fig:bayesian_graph} the variable $x_5$ depends on the values sampled from $x_3$ and $x_4$, and, in  turn, it affects the probability of which boolean we can sample from $x_6$. The network does not give a direct access to a vector
$(\ScalA_\True, \ScalA_\False)$ 
describing $x_5$ by its own,  
but only to a stochastic matrix $\StocA_5$ quantifying 
the conditional dependence of $x_5$ with respect to $x_3$ and $x_4$. Formally, $\StocA_5$ is a matrix with 
four rows, representing the four possible outcomes of a joint sample of $x_3$ and $x_4$ (\ie $(\True, \True)$,  $(\True, \False)$, $(\False, \True)$, $(\False, \False)$), and two columns, representing the two possible outcomes for $x_5$. The matrix is \emph{stochastic} in the sense that each line represents a probabilistic distribution of booleans: for instance, $(\StocA_5)_{(\True, \False), \True}=0.4$ and $(\StocA_5)_{(\True, \False), \False}=0.6$ mean that $\True$ can be sampled from $x_5$  with a 40\% chance, while $\False$ with 60\%, whenever $x_3$ has been observed to be $\True$ and $x_4$ to be $\False$. 

Having such a graph $\mathcal G$ and the stochastic matrices $\StocA_1$, $\StocA_2$, $\StocA_3$, \etc associated with its nodes, a typical query is to compute the joint probability of a subset of the variables of  $\mathcal G$. For example, the vector 
 $\Pr(x_3,x_6)=(\rho_{(\True, \True)},\rho_{(\True, \False)},\rho_{(\False, \True)},\rho_{(\False, \False)})$ giving the marginal over   $x_3$ and $x_6$, \ie 
the  probability of the possible outcomes of $x_3$ and $x_6$.   
A way to  obtain this is  first computing the joint distributions of all variables in the graph in a single shot and then summing out the variable  we are not interested in. This will give, for every possible boolean value $b_3,b_6$ in $\{\True,\False\}$ taken by, respectively, $x_3$ and $x_6$, the following expression:
\begin{equation}
\label{eq:let_term_semantics1}
	\sum_{
		\substack{
			b_1,b_2,b_4,
			b_5\in \{\True, \False\}
		}
	}
	(\StocA_1)_{b_1}
	(\StocA_2)_{b_1, b_2}
	(\StocA_3)_{b_2, b_3}
	(\StocA_4)_{b_4}
	(\StocA_5)_{(b_3, b_4), b_5}
	(\StocA_6)_{(b_2, b_5), b_6}\,.
\end{equation}
For each of  the $2^2$ possible values of the indexes $(b_3, b_6)$ we have a sum of $2^4$ terms. 
That is,   to compute the joint probability of $(x_3, x_6)$, {we have to compute $2^6$ entries.}
%
 This method is unfeasible in general as it requires a number of operations exponential in the size of $\mathcal G$. Luckily, one can take advantage of the conditional (in)dependencies underlined by $\mathcal G$ to get a better factorisation than in \eqref{eq:let_term_semantics1}, breaking the  computation in that of factors  of smaller size. For example:
\begin{equation}
\label{eq:let_term_semantics2}
	\sum_{b_5}\!
		\left(\!\sum_{b_2}\!
		\left(\!\sum_{b_1}
			(\StocA_1)_{b_1}
			(\StocA_2)_{b_1, b_2}
		\!\right)
		(\StocA_3)_{b_2, b_3}
		(\StocA_6)_{(b_2, b_5), b_6}
		\!\right)\!
		\left(\!
		\sum_{b_4}
		(\StocA_4)_{b_4}
		(\StocA_5)_{(b_3, b_4), b_5}
		\!\right)\!.
\end{equation}
Let us denote by $\FactA^i$ the intermediate factor in \eqref{eq:let_term_semantics2} identified by the sum over $b_i$: for example $\FactA^2$ is the sum $\sum_{b_2}\FactA^1_{b_2}(\StocA_3)_{b_2, b_3}(\StocA_6)_{(b_2, b_5), b_6}$. Notice that if we suppose to have memorised the results of  computing $\FactA^1$, to obtain  $\FactA^2$ {requires to compute $2^4$ entries}, \ie the cost is exponential in the number of the different indexes $b_j$'s appearing in the expression defining $\FactA^2$. By applying the same  reasoning to all factors in~\eqref{eq:let_term_semantics2}, 
	 one  notices that in the whole computation  {we never need to compute more than  $2^4$ entries}: we have gained a factor of $2^2$ with respect to \eqref{eq:let_term_semantics1}.

The Variable Elimination algorithm 
 performs factorisations like \eqref{eq:let_term_semantics2} in order to compute more efficiently the desired marginal distribution. The  factorisation is characterised by an ordered sequence of unobserved (or marginalised) variables to eliminate, \ie to sum out. The factorisation in \eqref{eq:let_term_semantics2} is induced by the sequence $(x_1,x_2,x_4,x_5)$:  $\FactA^1$ eliminates variable $x_1$, 
$\FactA^2$ then eliminates variable $x_2$, and so on. Different orders yield different factorisations with different performances, \eg{} the inverse order $(x_5,x_4,x_2,x_1)$ is 
less efficient, as the largest factor here requires to compute $2^5$ entries.

\paragraph{Factorising Inference, via the program structure.}

In the literature,  factorisations are usually described as collections of factors (basically vectors) and the \VE{} algorithm is presented as an iterative algorithm acting on such collections. In this paper we propose a framework that gives more \emph{structure} to this picture, expressing \VE{} as a program transformation, building a factorisation by  induction on the structure of a program, \emph{compositionally}. More precisely:
\begin{itemize}
\item we define a fragment $\LetCalculus$ of the linear simply-typed $\lambda$-calculus, which is  able to represent \emph{any factorisation} of a Bayesian network query as a $\lambda$-term. Random variables are associated with term variables of ground type;
\item we express the \VE{} algorithm as  rewriting over the $\lambda$-terms in $\LetCalculus$, consisting in reducing the scope of the variables that have to be eliminated. 
\end{itemize}
Our approach integrates and is grounded on the  denotational semantics of the terms, which   directly \emph{reflects and validates the factorisation} algorithm---yielding  soundness and completeness. We stress that inference is computing the semantics of the program (the marginal distribution defined by it).\\

The reader can easily convince herself that the query about the joint marginal distribution $(x_3, x_6)$ to the Bayesian network in Fig.~\ref{sub_fig:bayesian_graph} can be expressed by the let-term $\LetTerm$ in Fig.~\ref{sub_fig:let-term}, where we have enriched the syntax of $\lambda$-terms with the constants representing the stochastic matrices. We consider $\Let xe{e'}$ as a syntactic sugar for $(\lambda x. e')e$, so the term $\LetTerm$ can be seen as a $\lambda$-term, which is moreover typable in a linear type system like the one in Fig.~\ref{fig:terms_types}. The fact that some variables $x_i$ have more free occurrences in sub-terms of $\LetTerm$ is not in contrast with the linearity feature of the term, as let-expressions are supposed to be evaluated following a call-by-value strategy, and ground values (as e.g.~booleans) can be duplicated in linear systems (see Ex.~\ref{ex:typing} and Remark~\ref{rk:typing_as_MLL} for more details).

Terms of this type are  associated in \emph{quantitative} denotational semantics such as \cite{danosehrhard,LairdMMP13} with algebraic expressions which give the joint distribution of the output variables. 
%
%
Here, 
 in the same spirit as 	 \cite{fscd23},  we  adopt a  variant of the quantitative denotation 
 (Sect.~\ref{sect:factors}), which can be seen as a  compact reformulation of the original model, and is more suitable  to  deal with factorised inference.
It turns out that when we compositionally compute the semantics of $\LetTerm$ following the structure of the program, we have 
a more efficient computation than in \eqref{eq:let_term_semantics1}. This is because now \emph{the inductive interpretation 
 yields intermediate
factors of smaller size},  in a similar  way to what algorithms for exact inference do.  In the case of $\LetTerm$,  its inductive interpretation 
 behaves similarly to \VE{} given  the elimination order $(x_5,x_4,x_2,x_1)$, see Ex.~\eqref{ex:semantics_let}.

Notice that \emph{different programs} may encode  the \emph{same  model} and query, but  with a significantly \emph{different  inference cost}, due to their different structure. 
A natural question is then to wonder if we  can  directly act on the structure of the program, in such a way that \emph{the semantics is invariant, but inference is more efficient}. 
%
%
In fact, we show that the language $\LetCalculus$ is sufficiently expressive to represent \emph{all} possible factorisations of~\eqref{eq:let_term_semantics1}, \eg \eqref{eq:let_term_semantics2}. The main idea for such a representation arises from the observation that summing-out variables in the semantics corresponds in the syntax to make a let-in definition local to a sub-expression. For example, the factor $\FactA^1 = \sum_{b_1}(\StocA_1)_{b_1}(\StocA_2)_{b_1, b_2}$ of \eqref{eq:let_term_semantics2} can be easily obtained by making the variable $x_1$ local to the definition of $x_2$, creating a $\lambda$-term $e^1$ of the shape $\Let{x_1}{\StocA_1}{\StocA_2 x_1}$ and replacing the first two definitions of $\LetTerm$ with $\Let{x_2}{e^1}{\dots}$. In fact, the denotation of $e^1$ is exactly $\FactA^1$. What about $\FactA^2$? Here the situation is subtler as in order to make local the definition of $x_2$ one should gather together the definitions of  $x_3$ and $x_6$, but the definition of $x_6$ depends on a variable $x_5$ which in  turn depends on $x_3$, so a simple factor of ground types (\ie tensors of booleans) will generate a dependence cycle. Luckily, we can use a (linear) functional type, defining a $\lambda$-term $e^2$ as $\Let{x_2}{e^1}{\Seq{\StocA_3 x_2, \lambda x_5.\StocA_6 \Seq{x_2, x_5}}}$ and then transforming $\LetTerm$ into Fig.~\ref{sub_fig:let-term-VE}. Again, we can notice that the denotation of $e^2$ is exactly $\FactA^2$. Fig.~\ref{fig:example_reduction} details\footnote{Actually, the expressions $e^1$ and $e^2$ in Fig.~\ref{fig:example_reduction}  are a bit more cumbersome that the ones here discussed because of some bureaucratic let-in produced by a more formal treatment. This difference is inessential and can be avoided by adding a post-processing.} the whole rewriting mimicking the elimination of the variables $(x_1, x_2, x_4, x_5)$ applied to $\LetTerm$.  This paper shows how to generalise this reasoning to any let-term.


\paragraph{Contents of the paper.}
We present an algorithm which is able to \emph{fully} perform  \VE{}  \emph{on programs}: for  any elimination order, 
program $ \LetTerm $  is 
\emph{rewritten} by the   algorithm into program $ \LetTerm'$, representing---possibly in a more efficient way---the same  model. 	
As  stressed,  our investigation   is of \emph{foundational nature};
	we focus on a theoretical framework in which we are able to prove  the \emph{soundness and completeness} (Th.~\ref{th:soundnessVE_Let}, Cor.~\ref{cor:soundnessVE_Let_nary}) of the \VE{} algorithm on terms. To do so, we  leverage on the quantitative denotation of the terms.
The structure of the paper is as follows.

-- Sect.~\ref{sect:syntax} defines the linear $\lambda$-calculus $\LetCalculus$, and its semantics.  In particular, Fig.~\ref{fig:terms_types} gives the linear typing system, which is a fragment of multiplicative linear logic \cite{ll} (Remark~\ref{rk:typing_as_MLL}). Fig.~\ref{fig:denotation_terms} sketches the denotational semantics of $\LetCalculus$ as weighted relations \cite{LairdMMP13}.  At first, the reader who wishes to focus on \VE{} can skip the formal details about the semantics, and just 
read the intuitions  in Sect.~\ref{sec:gentle_semantics}.

-- Sect.~\ref{sect:factors} formalises the notion of factorisation as a set of factors (Def.~\ref{def:factor}) and shows how to associate a factorisation to the \emph{let-terms} in  $\LetCalculus$    (Def.~\ref{def:factor_of_let_term}). Def.~\ref{def:VE_factors}  recalls the standard \VE{} algorithm---acting on sets of factors---denoted by $\VEF$. 

-- Sect.~\ref{sect.VE_as_let_rewriting} is  the \emph{core of our paper},    capturing \VE{} as a let-term transformation (Def.~\ref{def:VE}), denoted here by $\VEL$, using the rewriting rules of Fig.~\ref{fig:let-rewriting}. We prove the correspondence between the two versions of $\VE$ in Th.~\ref{th:soundnessVE_Let} and Cor.~\ref{cor:soundnessVE_Let_nary}, stating our main result, the  soundness and completeness of  $\VEL$ with respect to $\VEF$.

 As an extra bonus
  (and a confirmation of the  robustness of our approach),   an enrichment of the semantics---based on  probabilistic coherence spaces \cite{danosehrhard}---allows us to prove   a nice property of the terms of $\LetCalculus$, namely that the total mass of their denotation is easily computable from the type of the terms (Prop.~\ref{prop:totality-mass}).

\subsection*{Related work }

 Variants of factorisation algorithms were invented independently in multiple communities (see \cite{KschischangFL01} for a survey). The algorithm of Variable Elimination (\VE{})  was first formalised in \cite{ZhangP94}. 
The approach to \VE{} which is usually taken by PPLs is  to compile a  program  into an \emph{intermediate  structure}, on which \VE{} is performed. 
Our specific  contribution is to provide the first algorithm which \emph{fully} performs \VE{} \emph{directly on  programs}.
As  explained, by this we mean the following. First, we observe that the inductive interpretation of a term behaves as  \VE{}, following an  ordering of the  variables to eliminate which is implicit in  the structure of the program---possibly a non-efficient one.  Second, our algorithm  transforms   the program in such a way that its structure reflects \VE{} according to any arbitrary ordering, while still  denoting the same model (the semantics is invariant).

As we discussed before, our  work builds on the programme put forward in   \cite{KollerMP97}. Pfeffer 
\cite{Pfeffer07}(page 417) summarizes this way the limits of the algorithm in \cite{KollerMP97}: "The 
the solution is only partial. Given a BN encoded in their language, the algorithm can
be viewed as performing Variable Elimination using \textit{a particular elimination
order}: namely, from the last variable in the program upward. It is well-known that
the cost of VE is highly dependent on the elimination order, so the algorithm is
exponentially more expensive for some families of models than an algorithm that
can use \emph{any order}."  The algorithm we present here achieves a full solution: any elimination order can be implemented.

The literature on  probabilistic programming languages and inference algorithms is vast, even restricting attention to \emph{exact inference}. At the beginning of the Introduction we have mentioned several relevant contributions.
 Here we briefly discuss  two lines of work which are especially relevant to our  approach. 

-- \emph{Rewriting  the program to improve inference efficiency, as in \cite{GorinovaGSV22}.}
 A key goal of PPLs is to separate the model description (the program) from the inference task. As pointed out by  \cite{GorinovaGSV22}, 
 such a  goal is hard to achieve in practice. To improve inference
 efficiency, users are often forced to re-write the program by hand.

 -- \emph{Exploiting the local structure of the  program to achieve efficient inference, as in \cite{HoltzenBM20}.} This  road is taken by the authors of the  language Dice---here the  algorithm   does not act on the program itself.
 


Our work incorporates elements from both lines: 
we perform inference  compositionally, following the inductive structure of the program; to improve efficiency, 
our rewriting algorithm modifies the program structure (modelling the  \VE{} algorithm),while  keeping the semantics   invariant. 

%
%
%

As a matter of fact, our first-order  language is very similar to the language Dice \cite{HoltzenBM20}, and has similar expressiveness.  In order to keep  presentation and proofs simple, we prefer to omit a conditioning construct such as   \texttt{observe}, but it  could easily be accommodated (see Sect.~\ref{sec:conclusions}). There are however significant differences.   We  focus on \VE,  while Dice  implements a different inference algorithm,  compiling programs  to weighted Boolean formulas, then performing  Weighted Model Counting \cite{ChaviraD05}. Moreover, as said, 
Dice exploits the local structure of the \emph{given} program, without program transformations to improve the inference cost, which is instead at the core of our approach.

Rewriting is  central to  \cite{GorinovaGSV22}. The  focus there is on 
	  probabilistic  programs with \emph{mixed}
	discrete and continuous parameters: by  eliminating the discrete parameters,  general  (gradient-based) algorithms can then be used.
	To \emph{automate} this process, the authors introduce an information flow type system that can detect conditional independencies;  rewriting  uses \VE{} techniques, even though the authors are not directly interested in the equivalence with the standard  \VE{} algorithm (this  is left as a conjecture).
In the same line, we mention also a very recent work \cite{LiWZ24} which tackle a similar task as  \cite{GorinovaGSV22}; while using  similar ideas,   a new  design in both the  language and the information flow type system allows the authors to deal with bounded recursion.  The term transformations have a different goal than ours, compiling a probabilistic  program into a pure one. However,  some key elements there resonate with our approach:  program transformations are based on continuation passing style (we use arrow variables in a similar fashion);  the language in  \cite{LiWZ24}
 is not defined by an operational semantics, instead  the authors  ---like us--- adopt a compositional, \emph{denotational} treatment.
 %

\paragraph{Denotational semantics versus  cost-awareness.} Our approach integrates and is grounded on  a quantitative \emph{denotational semantics}.
Pioneering work by \cite{JacobsZ16,JacobsZ} has  paved the way for a logical and semantical comprehension of Bayesian networks and inference from a categorical perspective, yielding an extensive body of work based on  the setting of string diagrams, \eg \cite{ChoJ19,JacobsKZ19,bookJacobs}. A denotational take on Bayesian networks is also at the core of \cite{Paquet21}, and underlies the categorical framework of \cite{Stein2021CompositionalSF}.
These lines of research however  
do not 
take into consideration  the \emph{computational cost}, which is the very reason motivating the introduction  and development  of Bayesian networks, and inference algorithms such as \VE.
In the literature, foundational  understanding tends to focus on either a compositional semantics or on efficiency, but the two worlds are separated, and typically explored as \emph{independent} entities. 
This dichotomy stands in stark contrast to  Bayesian networks, where the representation, the  semantics (\ie the underlying  joint distribution), and the inference algorithms are deeply \emph{intertwined}.
A new perspective  has been recently propounded  by \cite{fscd23},  advocating the need for a  quantitative semantical  approach more attentive to the resource consumption and to the actual cost of computing the semantics, which here exactly corresponds to performing inference. Our contribution fits in this line, which inspires  also the cost-aware semantics in  \cite{FaggianPV24}. The latter introduces a 
higher-order  language —in the idealized form of a $\lambda$-calculus—which is  sound and complete
w.r.t. Bayesian networks, together with a  type system which computes  the  cost of  (inductively performed) inference. 
Notice that   \cite{FaggianPV24} 
does not deal with terms transformations to rewrite a  program into a more efficient one. Such transformations, reflecting  the essence of the \VE\ algorithm,  is exactly the core  of our paper --- our algorithm  easily adapts to the first-order fragment of \cite{FaggianPV24}. 

\section{$\LetCalculus$ Calculus and Let-Terms}
\label{sect:syntax}

We consider a linear simply typed $\lambda$-calculus extended with stochastic matrices over tuples of booleans. Our results can be 
extended to more general systems, but here we focus on the core fragment able to represent Bayesian networks and
the factorisations produced by the \VE{} algorithm. In particular, we adopt  a specific class of types, where arrow types are restricted to  
\emph{linear} maps from (basically) tuples of booleans (i.e.~ the values of positive types in the grammar below) to pairs of a tuple of booleans and, possibly, another arrow.

\subsection{Syntax} 
We consider the following grammar of types:
\begin{align*}
\PosA, \PosB, \dots &\Bnf \Bool\mid\PosA\otimes\PosB&\text{(positive types)}\\
\NegA, \NegB, \dots&\Bnf\PosA\multimap\TypeA&\text{(arrow types)}\\
\TypeA, \TypeB, \dots &\Bnf\PosA\mid\NegA\mid\PosA\otimes\TypeA&\text{(let-term types)}
\end{align*}

It is convenient to adopt a typing system \emph{à la Church}, i.e.~ we will consider type annotated variables, meaning that we fix a set of variables and a function $\mathsf{ty}$ from this set to the set of positive and arrow types (see e.g.~\cite[ch.10]{Hindley2008}, this style is opposed to the typing system \emph{à la Curry}, where types should be associated to variables by typing contexts). 
We call a variable $v$ positive (resp.~arrow) whenever $\Type v$ is a positive (resp.~arrow) type. We use metavariables $x$, $y$, $z$ (resp.~$f$, $g$, $h$) to range over positive (resp.~arrow) variables. The letters $v$, $w$ will be used to denote indistinctly positive or arrow variables. 

The syntax of $\LetCalculus$ is given by the following 3-sorted grammar, where $\StocA$ is a metavariable corresponding to a stochastic matrix between tuples of booleans:
\begin{align*}
\vec v&\Bnf v \mid \Seq{\vec v, \vec v'}\quad\qquad\qquad\text{if $\FV{\vec v}\cap \FV{\vec v'}=\emptyset$}
&&&\text{(patterns)}
\\
e&\Bnf
	v
	\mid \StocA(\vec x)
	\mid \App f{\vec x}
	\mid \Seq{e,e'}
	\mid \Abs{\vec x}{e}
	\mid \Let{\vec v}{e}{e'}
&&&\text{(expressions)}
\\
\LetTerm&\Bnf 
	\vec v
	\mid \Let{\vec v}{e}{\LetTerm}
&&&\text{(let-terms)}	
\end{align*}
Notice that a pattern is required to have pairwise different variables.
We allow $0$-ary stochastic matrices, representing random generators of boolean tuples, which we will denote simply by $\StocA$, instead of $\StocA()$. We can assume to have two $0$-ary stochastic matrices $\True$ and $\False$ representing the two boolean values. 


We denote by $\FV e$ the set of free variables of an expression $e$. As standard, $\lambda$-abstractions $\Abs{\vec v}e$ and let-in $\Let{\vec v}{e'}{e}$ bind in the subexpression $e$ all occurrences of the variables in $\vec v$. 
 In fact, $\Let{\vec v}{e'}{e}$ can be thought as syntactic sugar for $(\lambda \vec v.e)e'$. Given a set of variables $\mathcal V$, we denote by $\mathcal V^a$ (resp.~$\mathcal V^+$) the subset of the arrow variables (resp.~positive variables) in $\mathcal V$, in particular $\FV[a] e$ denotes the set of arrow variables free in $e$. 

Patterns are a special kind of let-terms and these latter are a special kind of expressions. A pattern is called \textdef{positive} if all its variables are positive. We use metavariables $\vec x, \vec y, \vec z$ to range over positive patterns. A let-term is \textdef{positive} if its rightmost pattern is positive, \ie: $\Let{\vec v}{e}{\LetTerm}$ is positive if $\LetTerm$ is positive.

\begin{figure}
\centering
%
%
%
\def\defaultHypSeparation{\hskip .05in}
	\AxiomC{$\Type v$ positive or arrow}
     	\UnaryInfC{$v:\Type v$}
     	\DisplayProof
\quad
	\AxiomC{$f: \PosA\multimap\TypeA$}
	\AxiomC{$\vec x:\PosA$}
 	\BinaryInfC{$\App{f}\vec x:\TypeA$}
     	\DisplayProof
\quad
	\AxiomC{$\StocA:\PosA\multimap\PosB$}
	\AxiomC{$\vec x:\PosA$}
 	\BinaryInfC{$\StocA\vec x:\PosB$}
     	\DisplayProof

\medskip
	\AxiomC{$\vec x:\PosA$}
	\AxiomC{$e:\TypeA$}
 	\BinaryInfC{$\Abs{\vec x}e:\PosA\multimap\TypeA$}
     	\DisplayProof
\quad
	\AxiomC{$e:\PosA$}
	\AxiomC{$e':\TypeA$}
	\AxiomC{$\FV[a]{e}\cap\FV[a]{e'}=\emptyset$}
 	\TrinaryInfC{$\Seq{e,e'}:\PosA\otimes\TypeA$}
     	\DisplayProof

\medskip
\def\defaultHypSeparation{\hskip .05in}
	\AxiomC{$\vec v:\TypeA$}
	\AxiomC{$e:\TypeA$}
 	\AxiomC{$e':\TypeB$}
	\AxiomC{$\FV[a]{e}\cap\FV[a]{e'}=\emptyset$}
	\AxiomC{if $f\in\FV{\vec v}$ then $f\in\FV[a]{e'}$}
 	\QuinaryInfC{$\Let{\vec v}{e}{e'}:\TypeB$}
     	\DisplayProof
\caption{Typing rules: the binary rules suppose that the set of the free arrow variables of the subterms are disjoint; the let-rule binding an arrow variable $f$ requires also that this variable $f$ is free in the expression $e'$.}\label{fig:terms_types}
\end{figure}

Fig.~\ref{fig:terms_types} gives the rules generating the set of well-typed expressions (and so including patterns and let-terms). As standard in typing systems \emph{à la Church}, we  omit   an explicit typing environment of the  typing judgment $e:\TypeA$, as this can be recovered from the typing of the free variables of $e$, i.e.~if $\FV{e}=\{x_1,\dots,x_n\}$, then \emph{à la Curry} we would write $e:\TypeA$ by $x_1:\Type{x_1},\dots,x_n:\Type{x_1}\vdash e:\TypeA$. See Fig.~\ref{fig:typing} for an \emph{example of typing derivation} in both styles. 

The binary rules suppose the side condition $\FV[a]{e}\cap\FV[a]{e'}=\emptyset$ and the let-in rule binding an arrow variable $f$ has also the condition $f\in\FV[a]{e'}$. These conditions guarantee that arrow variables are used \emph{linearly}, ensuring the linear feature of the typing system.  
%

\begin{example}\label{ex:typing}
\begin{figure}
	\centering
	{\footnotesize 
		\begin{tabular}{l|l}
			\emph{Church style (our system):}	& \emph{Curry style:}  \\[2pt]
			\def\defaultHypSeparation{\hskip .05in}
			\AxiomC{$ v:P$ \quad $v':P$} 
			\AxiomC{$ v: P$}
			\AxiomC{$ v':P$}
			\BinaryInfC{$ (v,v'):P\otimes P$}
			\BinaryInfC{$  \Let{v'}{v}{(v,v')}:  P\otimes P$}
			\DisplayProof
			~
			&
			~
			\def\defaultHypSeparation{\hskip .05in}
			\AxiomC{$v:P\vdash v:P$}
			\AxiomC{$v:P\vdash v: P$}
			\AxiomC{$v':P\vdash v':P$}
			\BinaryInfC{$v: P, v':P \vdash (v,v'):P\otimes P$}
			\BinaryInfC{$ v: P \vdash \Let{v'}{v}{(v,v')}: P\otimes P$}
			\DisplayProof	
		\end{tabular}
	}	
	\caption{An example of type derivation, in both Church and Curry style.}	
	\label{fig:typing}	
\end{figure}	
Consider the term   $ \Let{v'}{v}{(v,v')}$, which will duplicate any  value assigned to  the free variable $v$.   If $v$ has \emph{positive type} (\eg~boolean),
it admits the  type derivation in Fig.~\ref{fig:typing}. 
On the contrast, if $v$ has \emph{arrow type}, no type derivation is possible, in agreement with the fact that arrows can only occur linearly. See also the discussion in Ex.~\ref{ex:cbv_beta}.
\end{example}

Notice that $\lambda$-abstractions are restricted to positive patterns. 
Also, a typing derivation of conclusion $e:\TypeA$ is completely determined by its expression. This means that  if an expression can be typed, then its type is unique. Because of that, we can extend the function $\mathsf{ty}$ to all expressions, \ie{} $\Type e$ is the unique type such that $e:\Type e$ is derivable, whenever $e$ is well-typed.

Notice that well-typed patterns $\vec v$ have at most one occurrence of an arrow variable (which is moreover in the rightmost position of the pattern). By extension of notation, we write $\vec v^a$ as the only arrow variable in $\vec v$, if it exists, otherwise we consider it as undefined. We also write by $\vec v^+$ for the pattern obtained from $\vec v$ by removing the arrow variable $\vec v^a$, if any. In particular $\vec x^+=\vec x$.

\begin{remark}\label{rk:ll_fragment}
The readers acquainted with linear logic  \cite{ll} may observe that the above grammar of types  identifies a precise fragment of this logic. In fact, the boolean type $\Bool$ may be expressed as the additive disjunction of the tensor unit: $\mathbf 1\oplus \mathbf 1$. Since $\otimes$ distributes over $\oplus$, positive types are isomorphic to $n$-ary booleans, for some $n\in\mathbb N$, \ie $\bigoplus_{n} \mathbf 1$, which is a notation for  $\mathbf 1\oplus\dots\oplus \mathbf 1$ $n$-times.


Moreover, by the isomorphisms $(\bigoplus_i\mathbf 1)\multimap\TypeA \simeq \bigwith_i(\mathbf 1\multimap\TypeA)\simeq\bigwith_i\TypeA$, where $\with$ denotes the additive conjunction, we deduce that the grammar of $\LetCalculus$ types is equivalent to an alternation of \emph{balanced} additive connectives, \ie it can be presented by the grammar:
$\TypeA \Def \mathbf 1 \mid \bigoplus_{n}\TypeA \mid \bigwith_{n}\TypeA$, for $n\in\mathbb N$.
The typing system hence identifies a fragment of linear logic which is not trivial, it has some regularity (alternation of the two additive connectives) and it is more expressive than just the set of arrows between tuples of booleans.
\end{remark}

\begin{remark}\label{rk:typing_as_MLL}
Notice that the binary rules might require to contract some positive types in the environment, as the expressions in the premises might have positive variables in common. In fact, it is well-known that 
contraction and weakening rules are derivable for the positive formulas in the environment, \eg $\mathbf 1\oplus\mathbf 1\vdash (\mathbf 1\oplus\mathbf 1)\otimes (\mathbf 1\oplus\mathbf 1)$ is provable in linear logic.  From a categorical point of view, this corresponds to the fact that positive types define co-algebras, and, operationally, that positive \emph{values} can be duplicated or erased without the need of being promoted (see \eg \cite{EhrhardT19} in the setting of PPLs).
\end{remark}


In the following we will represent a let-term $\LetTerm:=\Let{\vec v_1}{e_1}{\dots \Let{\vec v_n}{e_n}{\vec v_{n+1}}}$ by the more concise writing: 
$
	\LetTerm := \LetG{\vec v_1=e_1;\dots;\vec v_n=e_n}{\vec v_{n+1}}
$.

By renaming we can always suppose, if needed, that the patterns $\vec v_1,\dots, \vec v_{n}$ are pairwise disjoint sequences of variables and that none of these variables has occurrences (free or not) outside the scope of its binder. 

We call $\{\vec v_1=e_1$, \dots, $\vec v_n=e_n\}$ the set of the \textdef{definitions of $\LetTerm$} (which has exactly $n$ elements thanks to the convention of having $\vec v_1, \dots, \vec v_n$ pairwise disjoint), and $\biguplus_{i=1}^n\vec v_i$ the set of the \textdef{defined variables of $\LetTerm$}. The final pattern $\vec v_{n+1}$ is called the \textdef{output of $\LetTerm$}. Notice that $\LetTerm$ is positive if its output is positive. 

\begin{example}\label{ex:bayesian_graph_and_fact_as_let-term}
A Bayesian network of $n$ nodes can be represented by a closed let-term $\LetTerm$ having all variables positive and $n$ definitions of the form $x=\StocA(\vec y)$. The variables are associated with the edges of the graph and the definitions with the nodes such that $x=\StocA(y_1,\dots, y_k)$ represents a node with stochastic matrix $\StocA$, an outgoing edge associated with $x$ and $k$ incoming edges associated with, respectively, $y_1, \dots, y_k$. The output pattern contains the variables associated with a specific query to the Bayesian network. 

For instance, the Bayesian network in Fig.~\ref{sub_fig:bayesian_graph} is represented by the let-term:
$
	(x_1\!=\!\StocA_1; 
	x_2\!=\!\StocA_2x_1; 
	x_3\!=\!\StocA_3x_2; 
	x_4\!=\!\StocA_4; 
	x_5\!=\!\StocA_5\Seq{x_3, x_4};
	x_6\!=\!\StocA_6(x_2, x_5)
	\,	
	\mathtt{ in }
	\,
	\Seq{x_3, x_6}
	)
$,
which is the succinct notation for the let-term $\LetTerm$ in Fig.~\ref{sub_fig:let-term}. Notice that $\LetTerm$ induces a linear order on the nodes of the graph. The same graph can be represented by other let-terms, differing just from the order of its definitions. For example, by swapping the definitions of $x_3$ and $x_4$ we get a different let-term representing the same Bayesian network. We will consider this ``swapping'' invariance in full generality by defining the swapping rewriting $\Swap$ in Fig.~\ref{fig:let-rewriting} and stating Lemma~\ref{lemma:facts-invariant}. 

As discussed in the Introduction, let-terms with arrow variables and $\lambda$-abstractions might be needed to represent the result of applying the \VE{} algorithm to a Bayesian network. For instance, Fig.~\ref{fig:example_reduction} details the let-term produced by the elimination of the variables $(x_1, x_2, x_4, x_5)$. For example the closed subexpression $e_2$ in Fig.~\ref{fig:example_reduction} keeps local the variables $x_1$ and $x_2$ and has type $\Bool\otimes(\Bool\multimap\Bool)$. 
\end{example}

\subsection{Semantics}
\label{subsection:semantics}

\newcommand{\CoinEx}{\Coin{0.3}}
\newcommand{\Copy}{\mathtt{copy}}

We omit to detail an operational semantics of $\LetCalculus$, which can be defined in a standard way by using a sample-based or distribution-based semantics, in the spirit of \eg \cite{BorgstromLGS16}.
%
We prefer  to focus on the denotational semantics, which is more suitable to express the variable elimination algorithm in a compositional way. 
  Below, examples \ref{ex:cbv_linearity} and  \ref{ex:cbv_beta} informally  illustrate  the \emph{denotational} and \emph{operational} behavior of a $\mathtt{let}$-term, highlighting the \emph{linearity} of its nature.


%
%

\subsubsection{Semantics, a gentle presentation.}\label{sec:gentle_semantics}
We consider the semantics of
weighted relations \cite{LairdMMP13}, which is an example of quantitative semantics of linear logic interpreting programs as matrices over non-negative real numbers. 
The intuition behind this semantics is quite simple: each type $\TypeA$ is associated with a finite set $\Web\TypeA$ of indexes, called the \emph{web of $\TypeA$} (see \eqref{eq:web_types}). In case of a positive type, the web is the set of all possible outcomes of a computation of that type: the web of the boolean type $\Web\Bool$  is the set of the two booleans $\{\True,\False\}$, the web of a tensor $\PosA\otimes\PosB$ is the cartesian product $\Web{\PosA}\times\Web{\PosB}$ of the web of its components. An arrow type $\PosA\multimap\TypeA$ is also associated with the cartesian product $\Web\PosA\times\Web\PosB$, intuitively representing the elements of the trace of a function  of type $\PosA\multimap\TypeA$. To sum up, in this very simple fragment, webs are sets of nesting tuples of booleans, e.g.~$\Web{(\Bool\times\Bool)\multimap\Bool}=\{((b_1,b_2),b_3)\;\vert\;b_i\in\Web\Bool\}$ . 

The denotation $\Sem e$ of an expression $e$ is then a matrix (sometimes called weighted relation) whose rows are indexed by the sequences of the elements in the web of the free variables in $e$ and the columns are indexed by the elements in the web of the type of $e$. Eg, consider the expression $e_0$ given by $\Let{y}{fx}{(z,y)}$, with free variables $f:\Bool\multimap\Bool$, $x:\Bool$ and $z:\Bool$ and type $\Bool\times\Bool$. The matrix $\Sem{e_0}$ will have rows indexed by tuples $((b_1,b_2), b_3, b_4)$ and columns by $(b_5,b_6)$ for $b_i$'s in $\{\True,\False\}$. Intuitively, the entry $\Sem{e_0}_{((b_1,b_2), b_3, b_4),(b_5,b_6)}$ gives a weight to the possibility of a computation where the free variables of $e$ will ``behave'' as $(b_1,b_2)$ for $f$, $b_3$ for $y$ and $b_4$ for $z$ and the output will be $(b_5,b_6)$. 

The matrix $\Sem e$ is defined by structural induction on the expression $e$ by using matrix composition (for let-construction and application) and tensor product (for tuples), plus the diagonalisation of the indexes in the variables common to sub-expressions. Fig.~\ref{fig:denotation_terms} details this definition, giving a precise meaning to each programming construct. For example, taking the notation of Fig.~\ref{fig:denotation_terms}, the definition of $\Sem{(e',e'')}_{\overline a, (b',b'')}$ states that the weight of getting $(b',b'')$ supposing $\overline a$ is the product of the weights of getting $b'$ from $e'$ and $b''$ from $e''$, supposing $\overline a$ in both cases. The sharing of $\overline a$ in the two components of the tuple characterises the linearity of this calculus. Let us discuss this point with another example.

\begin{example}[Linearity, denotationally]
	\label{ex:cbv_linearity}	
	Let us write $\Coin {0.3} $ for a random generator of boolean values (a $0$-ary stochastic matrix), modeling a biased coin. In our setting $\Sem{\CoinEx}$ is a row vector $(0.3, 0.7)$ modeling the probability of sampling $\True$ or $\False$. 
	Let $e$ be the closed  term  $\Let{v}{\Coin {0.3}}{\Let{v'}{v}{(v,v')}}$, of type $\Bool\otimes\Bool$, well-typed because $v$ is positive ($\Type v = \Bool$). Since $e$ is closed, $\Sem{e}$ is also a row vector, now of dimension 4. One can easily check that 
	$\Sem{e}=(0.3, 0, 0, 0.7)$, stating that the only possible outcomes are the couples $(\True,\True)$ and $(\False,\False)$, while  $(\True,\False)$, $(\False,\True)$ have probability zero to happen.
	
	Notice that $\Sem{e}$ is different from $\Sem{(\CoinEx, \CoinEx)} = (0.3^2, 0.21, 0.21, 0.7^2)$. In fact, $\Sem{e}$ is \emph{linear} in $\CoinEx$, while $\Sem{(\CoinEx, \CoinEx)}$ is quadratic.
\end{example}

	


\begin{example}[Linearity and  $\mathtt{let}$-reduction] \label{ex:cbv_beta} 
	 Let us give an operational intuition for  the term $e$ in Ex.~\ref{ex:cbv_linearity}. There are two possibilities: we can  first sample a boolean from $\Coin {0.3}$ and then replace $v$ for the result of this sampling, or first replace $v$ for the sampler $\Coin {0.3}$, then sampling a boolean from each copy of $\Coin {0.3}$. The semantics states that we follow the former possibility and not the latter (as usual in a setting with effects). Intuitively,  $\Coin {0.3}$ reduces to a probabilistic sum 
	 $\blue{0.3} \, \True \blue{~+~ 0.7} \,\False$, and so $e$ first reduces to the sum \\
$ \blue{0.3}~\Let{v}{\True}{\Let{v'}{v}{(v,v')}}\blue{~+ ~ 0.7}~
   \Let{v}{\False}{\Let{v'}{v}{(v,v')}}$,\\
eventually yielding  $\blue{0.3} (\True,\True) \blue{~+~ 0.7} (\False,\False)$. 
In contrast,  duplicating  the sampler would yield $(\Coin {0.3},\Coin {0.3})$ whose  semantics is different, as discussed in Ex.~\ref{ex:cbv_linearity}.	 
	 Finally, notice    that replacing in $e$ the argument $\Coin {0.3}$ 
	 with an expression
	  $\lambda x.u$ (of arrow type) 
	 yields a term which is \emph{not  typable} in our  system (see  Ex.~\ref{ex:typing}). 
	\end{example} 

\todom{removed remark on value. We have no space and it is not very readable.  }

The rest of the subsection recalls from \cite{LairdMMP13} the definitions and notations of the denotational semantics, but the reader can jump to the next section if already satisfied with these intuitions and willing to focus on variable elimination.



\subsubsection{Semantics, formally.}
Let us  fix some basic notation from linear algebra. Metavariables $\SetA, \SetB, \SetC$ range over finite sets.
We denote by $\Size(\SetA)$ the cardinality of a set $\SetA$. We denote by  $\RP$ the cone of non-negative real numbers. Metavariables $\MatA,\MatB,\MatC$ will range over vectors in $\RP^{\SetA}$, for $S$ a \emph{finite} set, $\MatA_\ValA$ denoting the scalar associated  with $\ValA\in\SetA$ by $\MatA\in\RP^{\SetA}$. 
Matrices will be vectors indexed by pairs, \eg in $\RP^{\SetA\times\SetB}$ for $\SetA$ and $\SetB$ two finite sets. We may write $\MatA_{a,b}$ instead of $\MatA_{(a,b)}$ for $(a,b)\in\SetA\times\SetB$ if we wish to underline that we are considering indexes that are pairs. Given $\MatA\in\RP^{\SetA\times\SetB}$ and $\MatB\in\RP^{\SetB\times\SetC}$, the standard matrix multiplication is given by $\MProd{\MatA}{\MatB}\in\RP^{\SetA\times\SetC}$:\todom{R2: here the use of juxtaposition to indicate composition is potentially confusing because you are using a different order of composition than the usual function composition ($\circ$)  as discussed around lines 470-473.  If correct to do so, using the same symbol here (e.g. $\cdot$) might aid readability / comprehension.}
$(\MProd{\MatA}{\MatB})_{\ValA,\ValC}\Def\sum_{\ValB\in\SetB} \MatA_{\ValA,\ValB}\MatB_{\ValB,\ValC} \in \RP$.
The identity matrix is denoted $\delta\in\RP^{\SetA\times\SetA}$ and defined by $\delta_{a,a'}=1$ if $a=a'$, otherwise $\delta_{a,a'}=0$.

A less standard convention, but common in this kind of denotational semantics, is to consider the rows of a matrix $\MatA$ as the \emph{domain} and the columns as the \emph{codomain} of the underlined linear map. 
Hence, a vector in $\RP^{\SetA}$ is considered as a \emph{one line} matrix $\RP^{1\times\SetA}$, and the application of a vector $\MatB\in\RP^{\SetA}$ to a matrix $\MatA\in\RP^{\SetA\times\SetB}$, is given by $\MApp{\MatA}{\MatB}\Def\MProd{\MatB}{\MatA}\in\RP^{1\times\SetB}\cong\RP^{\SetB}$. 

The model denotes a type $\TypeA$ with a set $\Web\TypeA$, called the \textdef{web} of $\TypeA$, as follows:
\begin{align}\label{eq:web_types}
\Web{\Bool}&\Def\{\True,\False\}\,,
&
\Web{\PosA\otimes\TypeA}&\Def\Web{\PosA\multimap\TypeA}\Def\Web{\PosA}\times\Web{\TypeA}\,.
\end{align}
To denote an expression $e$, we must  associate a web with the set of free variables occurring in $e$. 
Given a finite set of variables $\mathcal V$, we define $\Web{\mathcal V}$ by using indexed products:
$
	\Web{\mathcal V}\Def\prod_{v\in\mathcal V}\Web{\Type v}
$.
Metavariables $\FValA, \FValB, \FValC$  denote elements in such webs $\Web{\mathcal V}$. In fact, $\FValA\in\Web{\mathcal V}$ can be seen as a function mapping any variable $v\in\mathcal V$ to an element $\FValA_v\in\Web{\Type v}$. We denote by $\star$ the empty function, which is the only element of $\Web{\emptyset}=\prod_{\emptyset}$.
Given a subset $\mathcal V'\subseteq\mathcal V$, we denote by 
$\Proj{\FValA}{\mathcal V'}$ the restriction of $\FValA$ to $\mathcal V'$, \ie $\Proj{\FValA}{\mathcal V'}\in\Web{\mathcal V'}$. Also, given two disjoint sets of variables $\mathcal V$ and $\mathcal W$ we denote by $\FValA\FPlus\FValB$ the union of an element $\FValA\in\Web{\mathcal V}$ and an element $\FValB\in\Web{\mathcal W}$, \ie $\FValA\FPlus\FValB\in\Web{V\uplus\mathcal W}$ and:
	$(\FValA\FPlus\FValB)_v\Def\FValA_v$ if $v\in\mathcal V$, and $(\FValA\FPlus\FValB)_v\Def\FValB_v$ if $v\in\mathcal W$.
%


\begin{figure}[t]
\begin{align*}
\Sem{v}_{\FValA, b}
&:=\delta_{\FValA, b}\\
\Sem{\Seq{e',e''}}_{\FValA, (b',b'')}
&:=\Sem{e'}_{\FValA\mid_{\FV{e'}},b'}\Sem{e''}_{\FValA\mid_{\FV{e''}},b''}\\
%
\Sem{\LetG{\vec v=e'}{e''}}_{\FValA, b}
&:=\textstyle\sum_{\FValC \in\Web{\vec v}}\Sem{e'}_{\FValA\vert_{\FV{e'}},\FValC}\Sem{e''}_{(\FValA\FPlus\FValC)\vert_{\FV{e''}},b}\\
\Sem{\Abs{\vec v}{e'}}_{\FValA, (b',b'')}
&:=\Sem{e'}_{\Proj{\FValA\FPlus b'}{\FV{e'}}, b''}\\
%
\Sem{\StocA(\vec x)}_{\FValA, b}
&:=\StocA_{\FValA, b}\\
\Sem{\App f{\vec x}}_{\FValA, b}
&:=\delta_{a',\FValA'''}\delta_{a'',b}&\hspace{-30pt}\text{where $\FValA\vert_f=(a', a'')$ and $\FValA\vert_{\vec x}=\FValA'''$}.
\end{align*}
\caption{Denotation of $e$ as a matrix $\Sem e$ giving a linear map from $\Web{\FV{e}}$ to $\Web{\Type e}$, so $\FValA\in\Web{\FV e}$ and $\ValB\in\Web{\Type e}$. In the tuple and $\lambda$ cases, we suppose $b=(b',b'')$.}\label{fig:denotation_terms}
\end{figure} 
An expression $e$ of type $\TypeA$ will be interpreted as a linear map $\Sem e$ from $\RP^{\Web{\FV{e}}}$ to $\RP^{\Web{\TypeA}}$. As such, $\Sem e$ can then be presented as a matrix in $\RP^{\Web{\FV{e}}\times\Web{\TypeA}}$. Fig.~\ref{fig:denotation_terms} recalls the definition of $\Sem e$ by structural induction on $e$. In the case of $\StocA(\vec x)$, we take the liberty to consider an element $\FValA\in\Web{\vec x}$ as actually the tuple of its components, ordered according to the order of the variables in the pattern $\vec x$. Similarly, when we compare $\FValA'''$ with $a'$ in $\Sem{\App f{\vec x}}$.

\begin{example}\label{ex:semantics_let}
Recall the term $\LetTerm$ in Ex.~\ref{ex:bayesian_graph_and_fact_as_let-term}. It is closed and of type $\Bool\otimes\Bool$, hence $\Sem{\LetTerm}$ is a one-row matrix in $\RP^{\Web{\emptyset}\times\Web{\Bool\otimes\Bool}}\simeq\RP^4$. By unfolding the definition in Fig.~\ref{fig:denotation_terms}, we get the following expression for $\Sem\LetTerm_{\star, (b_3, b_6)}$ with $b_3,b_6\in\{\True, \False\}$, where all $b_i$ vary over $\{\True, \False\}$, the index $i$ referring to the corresponding variable in $\LetTerm$: 
\begin{multline}
\label{eq:let_term_semantics}
\textstyle
	\sum_{b_1}(\StocA_1)_{b_1}
		\Biggl(\sum_{b_2}(\StocA_2)_{b_1, b_2}
			\biggl(\sum_{b'_3}(\StocA_3)_{b_2, b'_3}
				\Bigl(\sum_{b_4}(\StocA_4)_{b_4}
					\bigl(\sum_{b_5}(\StocA_5)_{(b_3,b_4), b_5}
\\[-10pt]
\textstyle
						\bigl(\sum_{b'_6}(\StocA_6)_{(b_2, b_5), b'_6}
						\delta_{b'_3, b_3}
						\delta_{b'_6, b_6}
						\bigr)
					\bigr)
				\Bigr)
			\biggr)
		\Biggr)						
\,.
\end{multline}
Expression~\eqref{eq:let_term_semantics} describes a way of computing $\Sem\LetTerm$ in a number of basic operations which is of order $2^3$ terms for each possible $2^2$ values of $b_3,b_6$.

For a more involved example, let us consider the let-term $\LetTerm'$ in line (L8) of Fig.~\ref{fig:example_reduction}, which is the result of the elimination of the variables $(x_1,x_2)$. We first calculate the semantics $\Sem{e_2}$ of the sub-expression keeping local $(x_1,x_2)$. Notice that  $e_2$ is a closed expression of type $\Bool\otimes(\Bool\multimap\Bool)$, so consider $b_3\in\Web{\Bool}$ and $(b_f,b'_f)\in\Web{\Bool\multimap\Bool}$, we have (after some simplification of $\delta$'s):
\begin{equation}
\label{eq:let_term_e_semantics}
\textstyle
\Sem{e_2}_{\star, (b_3, (b_f, b'_f))}=
	\sum_{b_2}
		\Bigl(
			\sum_{b_1}
			(\StocA_1)_{b_1}
			(\StocA_2)_{b_1,b_2}
		\Bigr)	
		(\StocA_3)_{b_2,b_3}	
		(\StocA_6)_{(b_2,b_f),b'_f}\,.
\end{equation}
We can then associate $\Sem{\LetTerm'}_{\star, (b_3, b_6)}$ with the following algebraic expression:
\begin{equation}
\label{eq:let_term'_semantics}
	\sum_{b_3', (b_f, b'_f)}\!\!\!\!
		\Sem e_{\star, (b_3, (b_f, b'_f))}\!
		\Biggl(\sum_{b_4}
				(\StocA_4)_{b_4}
				\Bigl(
					\sum_{b_5}
					(\StocA_5)_{b_4,b_5}
					\bigl(
						\sum_{b'_6}
						\delta_{b_5, b_f}
						\delta_{b_f',b'_6}					
					\bigr)
					\delta_{b_3',b_3}		
					\delta_{b'_6, b_6}	
					\bigr)\!
				\Bigr)\!\!
		\Biggr)
\end{equation}
Expression~\eqref{eq:let_term'_semantics} reduces to a number of basic operations which is of order $2^2$. By one memoizing the computation of $\Sem e$, Expression~\ref{eq:let_term'_semantics} offers a way of computing the matrix $\Sem{\LetTerm'}$ in a time linear in $2^2\times 2^2$. Indeed, Proposition~\ref{prop:let-rewriting_semantics} guarantees that $\LetTerm$ and $\LetTerm'$ (in fact all let-terms in Fig.~\ref{fig:example_reduction}) have the same denotational semantics: so the computation of $\Sem{\LetTerm'}$ gains a factor of $2$ with respect to \eqref{eq:let_term_semantics}. 
\end{example}

\newcommand\Ptdim[1]{\mathsf{dim}(#1)}
\newcommand\Theight[1]{\mathsf{ht}(#1)}

Let us conclude this subsection by observing that the type of a closed expression allows for computing the total mass of the denotational semantics of that expression.  With any positive type \(P\) we associate its dimension
\(\Ptdim P\in\Nat\) by \(\Ptdim\Bool=2\) and %
\(\Ptdim{P\otimes Q}=\Ptdim P\Ptdim Q\). This means that \(\Ptdim P\)
is the cardinality of \(\Web P\). And with any type \(T\) we associate
its height \(\Theight T\in\Nat\), the definition is: %
\(\Theight P=1\), \(\Theight{P\multimap T}=\Ptdim P\times\Theight
T\) %
and \(\Theight{P\otimes T}=\Theight T\).

\begin{proposition}
\label{prop:totality-mass}
  For any closed expression $e$, one has $\sum_{a\in\Web{\Type e}}\Sem e_{\star, a}=\Theight{\Type e}$.
\end{proposition}

\begin{example}\label{ex:non_stochastic_matrices}
Take the type $\Bool\otimes\Bool$ of the let-terms $\LetTerm$ and $\LetTerm'$ discussed in Example~\ref{ex:semantics_let}. We have that  $\Theight{\Bool\otimes\Bool}=1$, in accordance with the fact that all closed expressions of that type (such as $\LetTerm$ and $\LetTerm'$) describe joint probability distributions, so are denoted with vectors of total mass $1$. On the contrast, consider the type $\Bool\otimes(\Bool\multimap\Bool)$ of the expression $e_2$ keeping local the variables $x_1$ and $x_2$. We have $\Theight{\Bool\otimes(\Bool\multimap\Bool)}=\Theight{\Bool\multimap\Bool}=2$, which is the expected total mass of a stochastic matrix over booleans. However notice that the type $\Bool\otimes(\Bool\multimap\Bool)$ is subtler than that of a stochastic matrix $\Bool\multimap\Bool$: in fact, by using the isomorphisms discussed in Remark~\ref{rk:ll_fragment}, we have $\Bool\otimes(\Bool\multimap\Bool)\simeq(\Bool\multimap\Bool)\oplus(\Bool\multimap\Bool)$, which is the type of a probabilistic distribution of stochastic matrices.    
\end{example}


\section{Variable Elimination $\VEF$ over Let-Terms Factors}\label{sect:factors}

As mentioned in the Introduction, variable elimination is an iterative procedure transforming sets of factors (one can think of these as originally provided by a Bayesian network). We recall  this procedure, adapting it to our setting---in particular, we start from a set $\Facts{\LetTerm}$ of factors generated by a let-term $\LetTerm$ representing a Bayesian network. 
%
%
%
Subsect.~\ref{subsection:factors} defines factors and the main operations on them (product and summing-out). Subsect.~\ref{subsect:let-terms_sets_factors} shows how to associate a let-term $\LetTerm$ with a set of factors $\Facts{\LetTerm}$ such that from their product one can recover $\Sem\LetTerm$ (Prop.~\ref{prop:factors_as_semantics}). Finally, Subsect.~\ref{subsect:factorisation_let} presents the variable elimination algorithm as a transformation  $\VEF$ over $\Facts{\LetTerm}$ (Def.~\ref{def:VE_factors}) and Prop.~\ref{prop:order_sequence} gives the soundness of the algorithm. This latter result is standard from the literature (see e.g.~\cite{Darwiche2009}), and the contribution of this section is the definition of $\Facts{\LetTerm}$ which is essential to link this variable elimination $\VEF$ on factors  to our main contribution given in the next section: the variable elimination $\VEL$ as a term-rewriting process.

\subsection{Factors}
\label{subsection:factors}
%


\begin{definition}[Factor]\label{def:factor}
A \textdef{factor} $\FactA$  is a pair $(\FactVar\FactA, \FactFun\FactA)$ of a finite set $\FactVar\FactA$ of typed variables and a function $\FactFun\FactA$ from the web $\Web{\FactVar\FactA}$ to $\mathbb R_{\geq 0}$.

We will shorten the notation $\FactFun\FactA$ by writing just $\FactA$ when it is clear from the context that we are considering the function associated with a factor and not the whole pair $(\FactVar\FactA, \FactFun\FactA)$. 
We often consider $\FactFun\FactA$ as a vector indexed by the elements of its domain, so that ${\FactA}_{\FValA}$ stands for $\FactFun\FactA(\FValA)$, for every $\FValA\in\Web{\FactVar\FactA}$.

The \textdef{degree of $\FactA$}, written $\Deg_\FactA$, is the cardinality of $\FactVar\FactA$, and the  \textdef{base of $\FactA$}, written $\Base_\FactA$, is the maximal cardinality of $\Web{v}$ for every $v\in\FactVar\FactA$. Notice that $\Base_\FactA^{\Deg_\FactA}$ is an upper bound to the dimension of $\FactFun\FactA$, \ie the cardinality of $\Web{\FactVar\FactA}$. 
\end{definition}

\begin{example}\label{ex:factorM5}
Sect.~\ref{subsect:let-terms_sets_factors} formalises how to associate the definitions of a let-expression with factors. Let us anticipate a bit and see as an example the factor $\FactA$ that will be associated with the definition $x_5=\StocA_5(x_3,x_4)$ in the let-term in Ex.~\ref{ex:bayesian_graph_and_fact_as_let-term}. We have $\FactVar{\FactA}=\{x_3,x_4,x_5\}$ and for every $\ValA, \ValB, \ValC\in\Web{\Bool}$ we have $\FactFun\FactA(\ValA,\ValB,\ValC)=(\StocA_5)_{(\ValA,\ValB),\ValC}$.  Notice that $\FactA$ forgets the input/output (or rows/columns) distinction carried by the indexes of the stochastic matrix $\StocA_5$. 
\end{example}

A factor  $(\FactVar\FactA, \FactFun\FactA)$ involves two ``levels'' of indexing: one is given by the variables $v_1, v_2, \dots\in\FactVar\FactA$ tagging the different sets of the product $\Web{\FactVar\FactA}\Def\prod_{v\in\FactVar\FactA}\Web{v}$, and the other ``level'' is given by $\FValA, \FValB,\dots\in\Web{\FactVar\FactA}$ labelling the different components of the vector $\FactFun\FactA$, which we call web elements. 

Recall that the set of variables $\FactVar\FactA$ endows $\Web{\FactVar\FactA}$ with a cartesian structure, so that we can project a web element $\FValA\in\Web{\FactVar\FactA}$ on some subset of variables $\mathcal V'\subseteq\FactVar\FactA$ by writing $\Proj{\FValA}{\mathcal V'}$, as well as we can pair two web elements $\FValA\FPlus\FValA'$ whenever $\FValA\in\Web{\FactVar\FactA}$ and $\FValA'\in\Web{\FactVar\FactA'}$ and $\FactVar\FactA\cap\FactVar\FactA'=\emptyset$. 

\begin{figure}[t]
\begin{align*}
	\textstyle
	\FactVar{\FSum\CollA\FactA}
	& \Def \FactVar{\FactA}\setminus\CollA,
	&\FSum\CollA\FactA_{\FValA}
	&\textstyle\Def\sum_{\FValB\in\Val{\CollA}\cap\Val{\FactVar\FactA}}\FactA_{\FValA\FPlus\FValB}\\
	\FactVar{\FactA\FProd\FactB}
	& \Def \FactVar\FactA\cup\FactVar\FactB,
	&(\FactA\FProd\FactB)_{\FValC}
	&\!\Def\! \FactA_{\Proj{\FValC}{\FactVar\FactA}}\!\FactB_{\Proj{\FValC}{\FactVar\FactB}}
\end{align*}
\caption{\emph{Summing-out $\FSum\CollA\FactA$} of a set of variables $\CollA$ in a factor $\FactA$ and \emph{product $\FactA\FProd\FactB$} of two factors $\FactA$, $\FactB$. We suppose $\FValA\in\Web{\FactVar{\FactA}\setminus\CollA}$ and $\FValC\in\Web{\FactVar{\FactA\FProd\FactB}}$.}
\label{fig:summing-out_and_product}
\end{figure}
Fig.~\ref{fig:summing-out_and_product} defines the two main operations on factors: summing-out and binary products. We illustrate them with some examples and remarks. 

\begin{example}
By recalling the factor $\FactA$ of Ex.~\ref{ex:factorM5}, we have that $\FactVar{\FSum{\{x_3\}}\FactA}=\{x_4,x_5\}$ and for every $\ValA,\ValB\in\Web{\Bool}$, $\FSum{\{x_3\}}\FactA_{(\ValA,\ValB)}=\StocA_{(\True,\ValA),\ValB}+\StocA_{(\False,\ValA),\ValB}$. In fact, we can do weirder summing-out, as for example $\FactVar{\FSum{\{x_3, x_5\}}\FactA}=\{x_4\}$, so that $\FSum{\{x_3, x_5\}}\FactA_{\ValA}=\StocA_{(\True,\ValA),\True}+\StocA_{(\True,\ValA),\False}+\StocA_{(\False,\ValA),\True}+\StocA_{(\False,\ValA),\False}$ may be a scalar greater than one, no more representing a probability.
\end{example}


With the notations of Fig.~\ref{fig:summing-out_and_product}, if $\FactA$ is a join distribution over $\Web{\FactVar{\FactA}}$, the summing out of $\CollA$ in $\FactA$ gives the marginal distribution over $\Web{\FactVar{\FactA}\setminus\CollA}$. In the degenerate case where $\FactVar{\FactA}\subseteq\CollA$, then $\FactVar{\FSum\CollA\FactA}$ is the empty set and $\FSum\CollA\FactA_{\star}$ is the total mass of $\FactA$, \ie $\sum_{\FValB\in\Val{\FactVar{\FactA}}}\FactA_{\FValB}$.


\begin{example}\label{ex:Fprod_tensor}
Recall the factor $\FactA=(\{x_3,x_4,x_5\}, (\ValA,\ValB,\ValC\mapsto (\StocA_5)_{(\ValA,\ValB),\ValC}))$ of Ex.~\ref{ex:factorM5}, representing the definition $x_5=\StocA_5(x_3,x_4)$ in the let-term in Ex.\ref{ex:bayesian_graph_and_fact_as_let-term}, 
and consider a factor $\FactB=(\{x_3,x_4\}, (\ValA,\ValB\mapsto \StocA'_{\ValA,\ValB}))$
representing some definition $x_4=\StocA'(x_3)$. 
Then, $\FactVar{{\FactA\FProd\FactB}}=\{x_3,x_4,x_5\}$ and for every $\ValA,\ValB,\ValC\in\Web{\Bool}$, we have 
$\FactFun{\FactA\FProd\FactB}(\ValA,\ValB,\ValC)=(\StocA_5)_{(\ValA,\ValB),\ValC}\StocA'_{\ValA,\ValB}$. 
Notice that the factor product $\FactA\FProd\FactB$ is \emph{not} the tensor product $\otimes$ of the vectors $\FactFun{\FactA}$ and $\FactFun{\FactB}$, as variables can be shared between the different factors. In fact, the dimension of $\FactFun{\FactA}\otimes\FactFun{\FactB}$ is $2^3\times 2^2=2^{5}$, while $\FactFun{\FactA\FProd\FactB}$ is $2^3$. 
\end{example}

Notice that the computation of the sum out $\FSum\CollA\FactA$ is in $O\bigl(\Base_\FactA^{\Deg_\FactA}\bigr)$, as $\Base_\FactA^{\Deg_\FactA}$ is an upper bound to the cardinality of $\Val{\FactVar\FactA}$ which  gives the number of basic operations needed to define $\FSum\CollA\FactA$.
Analogously, the computation of $\FactA\FProd\FactB$ is in $O\bigl(\Base_{\FactA\FProd\FactB}^{\Deg_{\FactA\FProd\FactB}}\bigr)=O\bigl(\max(\Base_{\FactA},\Base_\FactB)^{\Deg_{\FactA}+\Deg_{\FactB}}\bigr)$, as $\Base_{\FactA\FProd\FactB}^{\Deg_{\FactA\FProd\FactB}}$ is an upper bound to the cardinality of $\Val{\FactA\FProd\FactB}$, which  gives the number of basic operations needed to define $\FactA\FProd\FactB$.

%
\begin{proposition}\label{prop:factor_prod_ass_distr}
Factor product is associative and commutative, with neutral element the empty factor $(\emptyset, 1)$. Moreover:
\begin{enumerate}
\item \label{fact_prop_prpr}
$\FSum{\CollA}{\FSum{\CollB}{\FactA}}=\FSum{\CollA\cup\CollB}{\FactA}$;
\item \label{fact_prop_prtimes}
$\FSum{\CollA}{\FactA\FProd\FactB}=(\FSum\CollA\FactA)\FProd\FactB$, whenever $\FactVar{\FactB}\cap\CollA=\emptyset$.
\end{enumerate}
%
\end{proposition}

\begin{definition}[$I$-factor product]\label{def:nary_factor_product}
Let $I$ be a finite set. Given a collection of factors $(\FactA_i)_{i\in I}$, we define their \textdef{factor product} as the factor $\BigFProd_{i\in I}\FactA_i\Def\FactA_{i_1}\odot\dots\odot\FactA_{i_n}$, for some enumeration of $I$. This is well-defined independently from the chosen enumeration because of Prop.~\ref{prop:factor_prod_ass_distr}.
\end{definition}
%
By iterating our remark on the complexity for computing binary products, we have that the computation of the whole vector $\FactVar{\BigFProd_{i\in I}\FactA_i}$ is in $O\bigl(\Size(I)\Base_{\BigFProd_{i\in I}\FactA_i}^{\Deg_{\BigFProd_{i\in I}\FactA_i}}\bigr)$, where we recall $\Size(I)$ denotes the cardinality of $I$.

\subsection{Let-terms as Sets of Factors}
\label{subsect:let-terms_sets_factors}


Let us introduce some convenient notation.  Metavariables $\FactSetA, \FactSetB, \FactSetC$ will range over finite sets of factors. We lift the notation for factors  to sets of factors: we write $\FactVar{\FactSetA}$ for the union $\bigcup_{\FactA\in\FactSetA}\FactVar{\FactA}$, so we can speak about a variable of $\FactSetA$ meaning a variable of one (or more) factor in $\FactSetA$; hence, the degree $\Deg_{\FactSetA}$ (resp.~the base $\Base_{\FactSetA}$) of $\FactSetA$ is the cardinality of  $\FactVar{\FactSetA}$ (resp.~the maximal cardinality of a set $\Web{v}$ for $v\in\FactVar{\FactSetA}$). Also, the operations of the sum-out and product with a factor are lifted component-wise, \ie  $\FSum\CollA{\FactSetA}\Def\{\FSum\CollA{\FactA}\st\FactA\in\FactSetA\}$ and $\FactB\FProd\FactSetA\Def\{\FactB\FProd\FactA\st\FactA\in\FactSetA\}$. In contrast, the $I$-factor product $\BigFProd\FactSetA$ returns the single factor result of the products of all factors in $\FactSetA$, according to Def.~\ref{def:nary_factor_product}.

Given a set of variables $\CollA$, it will be convenient to partition $\FactSetA$ into $\FactSetA_{\CollA}$ and $\FactSetA_{\neg\CollA}$, depending on whether a factor in $\FactSetA$ has common labels with $\CollA$ or not, \ie:
\begin{align}
\FactSetA_{\CollA}&\Def\{\FactA\in\FactSetA \st \FactVar\FactA\cap\CollA\neq\emptyset\}\,,&
\FactSetA_{\neg \CollA}&\Def\{\FactA\in\FactSetA \st \FactVar\FactA\cap\CollA=\emptyset\}\,.
\end{align}
Notice that $\FactSetA=\FactSetA_{\CollA}\uplus\FactSetA_{\neg\CollA}$, as well as $\FactVar{\FactSetA}\cap\CollA\subseteq\FactVar{\FactSetA_{\CollA}}$ and $\FactVar{\FactSetA_{\neg\CollA}}\subseteq\FactVar{\FactSetA}\setminus\CollA$. 
In the case of singletons $\{v\}$, we can simply write $\FactSetA_{v}$ and $\FactSetA_{\neg v}$. 

\begin{definition}[$\Fact{\vec v = e}$]\label{def:factor_of_def}
Given a pattern $\vec v$ and expression $e$ s.t.~$\FV{\vec v} \cap \FV{e}=\emptyset$, we define $\Fact{\vec v = e}$, by: $\FactVar{\Fact{\vec v = e}}\Def\FV{e}\uplus\FV{\vec v}$ and 
$\FactFun{\Fact{\vec v = e}}\Def \FValA\FPlus\FValB \mapsto \Sem{e}_{\FValA, \FValB}$, for $\FValA\in\Web{\FV{e}}$, $\FValB\in\Web{\FV{\vec v}}$.
\end{definition}
In a definition $\vec v = e$, $e$'s free variables can be seen as input channels, while $\vec v$'s variables as output channels. This is also reflected in the matrix $\Sem{e}$ where rows are associated with inputs and columns with outputs. In contrast, a factor forgets such a distinction, mixing all indexes in a common family. 

Let us warn that Def.~\ref{def:factor_of_def} as well as the next Def.~\ref{def:factor_of_let_term} are not compatible with renaming of bound variables (a.k.a.~$\alpha$-equivalence), as they use bound variables as names for the variables of factors. Of course, one can define an equivalence of factors by renaming their variables, but this must be done consistently on all factors taken in consideration.

\begin{definition}[$\Facts{\LetTerm}$]\label{def:factor_of_let_term}
Given a let-term $\LetTerm$ with output pattern $\vec w$, we define the set of factors $\Facts{\LetTerm}$, by induction on the number of definitions of $\LetTerm$:
\begin{align*}
\Facts{\vec w}
&\Def
(\FV{\vec w}, \vec\ValA\mapsto 1)
\\
\Facts{\LetG{\vec v=e}{\LetTerm}}
&\Def
\begin{cases}
	\left\{\FSum{f}{\Fact{\vec v = e}\!\FProd\!\Facts{\LetTerm}_{f}}\right\}\!\uplus\!\Facts{\LetTerm}_{\neg f}
	&\text{if $f\!\in\!\FV[a]{\vec v}\!\setminus\!\FV{\vec w},$}\\[10pt]
	\{\Fact{\vec v = e}\}\uplus\Facts{\LetTerm}	
	&\text{otherwise}.
\end{cases}
\end{align*}


\end{definition}
\noindent
The definition of $\Facts{\LetG{\vec v=e}{\LetTerm}}$ is justified by the linear status of the arrow variables, assured by the typing system. In a let-term $\LetG{\Seq{\vec x, f}=e}{\LetTerm}$, we have two disjoint cases: either the arrow variable $f$ occurs free exactly once in one of the definitions of $\LetTerm$, or $f$ is free in the output $\vec w$ of $\LetTerm$. In the former case, $\Facts{\LetTerm}_{f}$ is a singleton $\{\FactA\}$, and we can sum-out $f$ once multiplied $\Fact{\Seq{\vec x, f} = e}$ with $\FactA$, as no other factor will use $f$. In the latter case, we keep $f$ in the family of the factors associated with the let-term, as this variable will appear in its output. 

%

\begin{example}\label{ex:factors_let_terms}
Let us consider  
Fig.~\ref{fig:example_reduction}. The let-term $\LetTerm$ in (L1) has exactly $7$ factors, the $1$-constant factor associated with the output and one factor for each definition, carrying the corresponding stochastic matrix $\StocA_i$. For a less obvious example, consider the term $\LetTerm'$ in (L8).
The set $\Facts{\LetTerm'}$ has 4 factors: one for the output, two associated with the definitions of, respectively, $x_4$ and $x_5$ and the last one defined as $\FSum{f}{\Fact{x_3, f = e_2}\FProd\Fact{x_6 = f x_5}}$. Notice that $\Fact{x_6 = f x_5}_{\FValA}=1$ if $\FValA_f = (\FValA_{x_5}, \FValA_{x_6})$ otherwise  $\Fact{x_6 = f x_5}_{\FValA}=0$. Therefore the sum-out on $f$ produces a sum of only one term, whenever fixed $b_5\in\Web{x_5}$ and $b_6\in\Web{x_6}$.

Notice also that all let-terms from line (L12) have a set of factors of cardinality two, although they may have more than one definition. 
\end{example}

The following proposition shows how to recover the quantitative semantics $\Sem{\LetTerm}$  of a let-term $\LetTerm$ out of the set of factors $\Facts{\LetTerm}$: take the product of all factors in $\Facts{\LetTerm}$ and sum-out all variables that are not free in $\LetTerm$ nor occurs in the output. The proposition is proven by induction on $\LetTerm$.  See Appendix~\ref{proof:factors}.
\begin{proposition}\label{prop:factors_as_semantics}
Consider a let-term $\LetTerm$ with output $\vec v$. Let $\mathcal F=\FactVar{\Facts\LetTerm}$, and consider $\FValA \in\Web{\FV{\LetTerm}}$, $\FValB\in\Web{\FV{\vec v}}$. If $\Proj{\FValA}{\FV{\LetTerm}\cap\FV{\vec v}}=\Proj{\FValB}{\FV{\LetTerm}\cap\FV{\vec v}}$, with $\FValA' = \Proj\FValA{\FV{\LetTerm}\setminus \FV{\vec v}}$,
$\FValB' = \Proj\FValB{\FV{\vec v}\setminus \FV{\LetTerm}}$, and
$\FValC=\Proj\FValA{\FV{\LetTerm}\cap\FV{\vec v}}\!\!\!=\!\Proj\FValB{\FV{\LetTerm}\cap\FV{\vec v}}$, we have 
$\Sem{\LetTerm}_{\FValA,\FValB} =\FSum{\mathcal F\setminus(\FV\LetTerm\cup\FV{\vec v})}{\BigFProd\Facts{\LetTerm}}(\FValA' \FPlus \FValC \FPlus \FValB')$. Otherwise $\Sem{\LetTerm}_{\FValA,\FValB} =0$.
In particular, if $\LetTerm$ is closed, then  $\Sem{\LetTerm}_{\star, \FValB} = \FSum{\mathcal F\setminus\vec v}{\BigFProd\Facts{\LetTerm}}(\FValB)$.
\end{proposition}

\subsection{Variable Elimination $\VEF$ over Sets of Factors}\label{subsect:factorisation_let}

We recall the definition of the variable elimination algorithm as acting on sets of factors. Prop.~\ref{prop:order_sequence} states its soundness, which is a standard result that we revisit here just to fix our notation. We refer to \cite[ch.6]{Darwiche2009} for more details. 

\begin{definition}[Variable elimination over sets of factors]\label{def:VE_factors}
The elimination of a variable $v$ in a set of factors $\FactSetA$ is the set of factors $\VEF(\FactSetA, \LabelA)$ defined by:
\begin{align}
	\VEF(\FactSetA, v)&\Def\{ \textstyle\sum_v\BigFProd\FactSetA_{v}\}\uplus\FactSetA_{\neg v}
\end{align}
This definition extends to finite sequences of variables $\Seq{v_1,\dots, v_h}$ by iteration:
\begin{align}
	\label{eq:ve_iterative}
	\VEF(\FactSetA, \Seq{v_1,\dots, v_h})&\Def \VEF(\VEF(\FactSetA, v_1), \Seq{v_2,\dots, v_h})
\end{align}
if $h>0$, otherwise $\VEF(\FactSetA, \Seq{})=\FactSetA$.
\end{definition}

\begin{example}\label{ex:VE_factors}
Recall the sets of factors $\Facts\LetTerm$ and $\Facts{\LetTerm'}$ of Ex.~\ref{ex:factors_let_terms}. An easy computation gives: $\Facts{\LetTerm'}=\VEF(\Facts\LetTerm, (x_1,x_2))$.
\end{example}

The soundness of  $\VEF(\FactSetA, \Seq{v_1,\dots, v_h})$ follows by induction on  the length $h$ of the sequence $\Seq{v_1,\dots, v_h}$, using  Prop.~\ref{prop:factor_prod_ass_distr}  (see Appendix \ref{proof:factors}):
\begin{proposition}\label{prop:order_sequence}
We have: $\BigFProd\VEF(\FactSetA, \Seq{v_1,\dots, v_h}) = \sum_{\{v_1,\dots, v_h\}}\BigFProd\FactSetA$.
In particular, $\FactVar{\VEF(\FactSetA, \Seq{v_1,\dots, v_h})} = \FactVar{\FactSetA}\setminus\{v_1,\dots, v_h\}$.
\end{proposition}

The above soundness states that the $\VEF$ transformation corresponds to summing-out the variables to eliminate from the product of the factors taken into consideration. This means that if the factors in $\Gamma$ represent random variables, then $\BigFProd\VEF(\FactSetA, \Seq{v_1,\dots, v_h})$ computes the join distribution over the variables in $\FactVar\Gamma\setminus\Seq{v_1,\dots, v_h}$.

%


\section{Variable Elimination $\VEL$ as Let-Term Rewriting}
\label{sect.VE_as_let_rewriting}

This section contains our main contribution, expressing the variable elimination algorithm syntactically, as a rewriting of let-terms, transforming the ``eliminated'' variables from global variables (\ie defined by a definition of a let-term and accessible to the following definitions), into local variables (\ie private to some subexpression in a specific definition). 
Subsect.~\ref{subsection:let-rewriting} defines such a rewriting $\Red$ of let-terms (Fig.~\ref{fig:let-rewriting}) and states some of its basic properties. Subsect.~\ref{subsect:ve_strategy} introduces the $\VEL$ transformation as a deterministic strategy to apply $\Red$ in order to make local the variable to be eliminated (Def.~\ref{def:VE}), without changing the denotational semantics of the term (Prop.~\ref{prop:let-rewriting_semantics}). Theorem~\ref{th:soundnessVE_Let} and Corollary~\ref{cor:soundnessVE_Let_nary} prove that $\VEL$ and $\VEF$ are equivalent, showing that $\Facts{\cdot}$ commutes over the two transformations. Finally, Subsect.~\ref{subsect:complexity} briefly discusses some complexity properties, namely that the $\VEL$ increases the size of a let-term quite reasonably, keeping a linear bound. 

\subsection{Let-Term Rewriting}\label{subsection:let-rewriting}
\begin{figure}
\begin{align*}
\text{($\Swap_1$)}&&
\LetG{\vec v_1=e_1;\vec v_2=e_2}{\LetTerm}&
	\Red\LetG{\vec v_2=e_2;\vec v_1=e_1}{\LetTerm}\\
&&	&\text{if $\FV{\vec v_1}\cap\FV{e_2}=\emptyset$,}\\[5pt]
\text{($\Swap_2$)}&&
\LetG{\vec v_1=e_1;\vec v_2=e_2}{\LetTerm}&	
	\Red\LetG{
		f=\Abs{\vec x}{e_2} ;
		\vec v_1 = e_1;
		\vec v_2 = \App{f}{\vec x}
		}{\LetTerm}\\
&&	&\text{if $\vec x = \FV{\vec v_1}\cap\FV{e_2}$ positive and not empty,}\\[5pt]
\text{($\Swap_3$)}&&
\LetG{\vec v_1 = e_1 ; \vec v_2 = e_2}{\LetTerm}&
	\Red\LetG{\Seq{\vec v^+_1, \vec v_2} = \LetG{\vec v_1=e_1}{\Seq{\vec v^+_1, e_2}}}
	{\LetTerm}\\
&&	&\text{if ${\vec v^a_1} = f$, with $f\in\FV{e_2}$,}\\[5pt]
\text{($\Mult$)}&&
\LetG{\vec v_1=e_1 ; \vec v_2 = e_2}{\LetTerm}&	
	\Red\LetG{\Seq{\vec v_1,\vec v_2} = 
	\LetG{\vec v_1 = e_1}{\Seq{\vec v_1, e_2}}}{\LetTerm}\\
&&	&\text{if $\vec v_1$ positive,}\\[5pt]
\text{($\Elimin_x$)}&&
\LetG{\vec v=e_1}{\LetTerm}&	
	\Red\LetG{\vec v'=\LetG{\vec v=e_1}{\vec v'}}{\LetTerm}\\
&&	&\text{if $x\notin\FV{\LetTerm}$ and $\vec v'$ is not empty and removes $x$ in $\vec v$.}
\end{align*}
\caption{Let-terms rewriting rules. We recall that $x$'s variables ($f$'s variables) are supposed positive (resp.~arrow), while $v$'s may be positive or arrow. We also recall from Section~\ref{sect:syntax} that $\vec v^a$ denotes the only arrow variable in a pattern $\vec v$, if it exists, and $\vec v^+$ denotes the pattern obtained from $\vec v$ by removing the arrow variable $\vec v^a$, if any. In the case $\vec v^+$ is empty, the notation $\Seq{\vec v^+, e}$ stands for $e$.}\label{fig:let-rewriting}
\end{figure}
Fig.~\ref{fig:let-rewriting} gives the rewriting rules of let-terms that we will use in the sequel. The rewriting steps $\Swap_1,\Swap_2,\Swap_3$ are called \emph{swapping} and we write $\LetTerm\Red[\Swap]\LetTerm'$ whenever $\LetTerm'$ is obtained from $\LetTerm$ by applying any such swapping step. The rewriting step $\Mult$ is called \emph{multiplicative} and it is used to couple two definitions. The reason why $\Swap_3$ is classified as swapping rather than multiplicative reflects the role of arrow variables in the definition of $\Facts{\LetTerm}$. Finally, the rewriting step $\Elimin_x$ \emph{eliminates} a positive variable $x$ from the outermost definitions, supposing this variable is not used in the sequel. The conditions in each rule guarantee that the rewriting $\Red$ preserves typing as stated by the following proposition (see Appendix~\ref{proof:let-rewriting}).

\begin{proposition}[Subject reduction]\label{prop:let_rewriting_subject}
The rewriting $\Red$ of  Fig.~\ref{fig:let-rewriting} preserves typing, \ie if $\LetTerm\Red\LetTerm'$ and $\LetTerm$ is of type $\TypeA$, then so is $\LetTerm'$, as well as $\FV{\LetTerm}=\FV{\LetTerm'}$.
\end{proposition}

\begin{proposition}[Semantics invariance]\label{prop:let-rewriting_semantics}
The rewriting $\Red$ of Fig.~\ref{fig:let-rewriting} preserves the denotational interpretation, \ie if $\LetTerm\Red\LetTerm'$ then $\Sem\LetTerm=\Sem{\LetTerm'}$.
\end{proposition}

Moreover, $\Facts{\LetTerm}$ is invariant under commutative rewriting (Appendix~\ref{proof:let-rewriting}):
\begin{lemma}\label{lemma:facts-invariant}
If $\LetTerm\Red[\Swap]\LetTerm'$, then $\Facts{\LetTerm'}=\Facts{\LetTerm}$.
\end{lemma}

\subsection{Variable Elimination Strategy}\label{subsect:ve_strategy}

The $\VEL$ transformation can be seen as a deterministic strategy of applying the rewriting $\Red$ in order to make local a variable in a let-term. The idea of $\VEL(\LetTerm,x)$ is the following: first, we gather together of definitions $(\vec v_i = e_i)$ of $\LetTerm$ having $x$ free in $e_i$ into a common huge definition $\vec v = e$ and we move this latter close to the definition of $x$ in $\LetTerm$; then, we make the definition of $x$ local to $e$. To formalise this rewriting sequence we define two auxiliary transformations: the swapping definitions $\SD$ (Def.~\ref{definition:SwapDef}) and the variable anticipation $\VA$ (Def.~\ref{definition:VA}).  

The \textdef{swapping definition} procedure rewrites a let-term $\LetTerm$ with at least two definitions by swapping (or gathering) the first definition with the second one, without changing the factor representation\SLV{}{(Lemma~\ref{lemma:SD_soundness})}.

\begin{definition}[Swapping definitions]\label{definition:SwapDef}
We define $\SD(\LetTerm)$ for a let-term $\LetTerm\Def\LetG{\vec v_1=e_1, \vec v_2=e_2}{\LetTerm'}$ with at least two definitions. The definition splits in the following cases, depending on the dependence of $e_2$ with respect to $\vec v_1$. 
\SLV
{
\begin{enumerate}
	\item If $\FV{\vec v_1}\cap\FV{e_2} = \emptyset$,
		$
			\SD(\LetTerm)\Def \LetG{\vec v_2=e_2; \vec v_1=e_1}{\LetTerm'}$.
	\item If $\FV{\vec v_1}\cap\FV{e_2}=\vec x$ is a non-empty sequence of positive variables, $
		\SD(\LetTerm)\Def \LetG{g=\Abs{\vec x}{e_2}; \vec v_1=e_1; \vec v_2=\App g{\vec x}}{\LetTerm'}$.
	\item If $\vec v^a_1 = f$ and $f \in \FV{e_2}$,
		$
					\SD(\LetTerm)\Def \LetG{\Seq{\vec v^+_1,\vec v_2}=
				\LetG{\vec v_1=e_1}
				{
					\Seq{\vec v^+_1, e_2}}
				}{\LetTerm'}$,
	if $\vec v^+_1$ is non-empty, otherwise: $\SD(\LetTerm)\Def \LetG{\vec v_2=\LetG{\vec v_1=e_1}{e_2}}{\LetTerm'}$.
\end{enumerate}
}
{
\begin{enumerate}
	\item If $\FV{\vec v_1}\cap\FV{e_2} = \emptyset$,
		\[
			\SD(\LetTerm)\Def \LetG{\vec v_2=e_2; \vec v_1=e_1}{\LetTerm'}.
		\]
	\item If $\FV{\vec v_1}\cap\FV{e_2}=\vec x$ is a non-empty sequence of positive variables,
		\[
			\SD(\LetTerm)\Def \LetG{g=\Abs{\vec x}{e_2}; \vec v_1=e_1; \vec v_2=\App g{\vec x}}{\LetTerm'}.
		\]
	\item If $\vec v^a_1 = f$ and $f \in \FV{e_2}$,
		\[
					\SD(\LetTerm)\Def \LetG{\Seq{\vec v^+_1,\vec v_2}=
				\LetG{\vec v_1=e_1}
				{
					\Seq{\vec v^+_1, e_2}}
				}{\LetTerm'},
		\]
	if $\vec v^+_1$ is non-empty, otherwise: $\SD(\LetTerm)\Def \LetG{\vec v_2=\LetG{\vec v_1=e_1}{e_2}}{\LetTerm'}$.
\end{enumerate}
}
\end{definition}
Notice that the above cases are exhaustive. In particular, if $\vec v_1$ has some variables in common with $\FV{e_2}$ then either all such common variables are positive or one of them is an arrow variable $f$. By case inspection and Lemma~\ref{lemma:facts-invariant}, we get:

\begin{lemma}[$\SD$ soundness]\label{lemma:SD_soundness}
Given a let-term $\LetTerm$ with at least two definitions, then $\LetTerm\Red[\Swap]\SD(\LetTerm)$, for the swap reduction $\Swap$ defined in Fig.~\ref{fig:let-rewriting}. In particular, $\SD(\LetTerm)$ is a well-typed let-term having the same type of $\LetTerm$ and such that $\Facts{\LetTerm}=\Facts{\SD(\LetTerm)}$.
\end{lemma}


Given a set of variables $\mathcal V$, the \emph{variable anticipation} procedure rewrites a let-term $\LetTerm$ into $\VA(\LetTerm, \mathcal V)$  by ``gathering'' in the first position all definitions having free variables in $\mathcal V$ or having arrow variables defined by one of the definitions already ``gathered''. This definition is restricted to positive let-terms. 

\begin{definition}[Variable anticipation]\label{definition:VA}
We define a let-term $\VA(\LetTerm, \mathcal V)\Def\LetG{\vec v'=e'}{\LetTerm'}$, given a positive let-term $\LetTerm\Def\LetG{\vec v_1=e_1}{\LetTerm_1}$ with at least one definition and a set of variables $\mathcal V\subseteq\FV{\LetTerm}$ disjoint from the output variables of $\LetTerm$. The definition is by structural induction on $\LetTerm$ and splits in the following cases. 
\SLV{
\begin{enumerate}
	\item If $\mathcal V=\emptyset$, then define:
		$
			\VA(\LetTerm, \mathcal V)\Def\LetTerm
		$.
		
	\item If $\mathcal V\cap\FV{e_1}=\emptyset$, so that $\mathcal V\subseteq\FV{\LetTerm_1}$, then define:\\
		$
			\VA(\LetTerm, \mathcal V)\Def\SD(\LetG{\vec v_1=e_1}{\VA(\LetTerm_1, \mathcal V)})$.			
	
	\item If $\mathcal V\cap\FV{e_1}\neq \emptyset$ and  $\vec v_1$ is positive, then consider $\VA(\LetTerm_1, \mathcal V\cap\FV{\LetTerm_1})\Def\LetG{\vec v'=e'}{\LetTerm'}$ and set:
		$
			\VA(\LetTerm, \mathcal V)\Def
			\LetG{\Seq{\vec v_1,\vec v'}=\LetG{\vec v_1=e_1}{\Seq{\vec v_1,e'}}}{\LetTerm'}$.

	\item If $\mathcal V\cap\FV{e_1}\neq\emptyset$ and $\vec v^a_1=f$. Notice that, by hypothesis, $f$ does not appear in the output of $\LetTerm_1$, as $\LetTerm$ (and hence $\LetTerm_1$) is positive. So we can consider $\VA(\LetTerm_1, (\mathcal V\cap\FV{\LetTerm_1})\cup\{f\})\Def\LetG{\vec v'=e'}{\LetTerm'}$ and define:
		$
			\VA(\LetTerm, \mathcal V)\Def
			\LetG{
				\Seq{\vec v^+_1,\vec v'} =
				\LetG{\vec v_1 = e_1}{\Seq{\vec v_1^+, e'}}
			}
			{\LetTerm'}$,
		\noindent if $\vec v^+_1$ is non-empty, otherwise: $\VA(\LetTerm, \mathcal V)\Def \LetG{\vec v'=\LetG{\vec v_1=e_1}{e'}}{\LetTerm'}$.
\end{enumerate}
}{
\begin{enumerate}
	\item If $\mathcal V=\emptyset$, then define:
		\[
			\VA(\LetTerm, \mathcal V)\Def\LetTerm\,.
		\]
	\item If $\mathcal V\cap\FV{e_1}=\emptyset$, so that $\mathcal V\subseteq\FV{\LetTerm_1}$, then define:
		\[
			\VA(\LetTerm, \mathcal V)\Def\SD(\LetG{\vec v_1=e_1}{\VA(\LetTerm_1, \mathcal V)})\,.
		\]				
	
	\item If $\mathcal V\cap\FV{e_1}\neq \emptyset$ and  $\vec v_1$ is positive, then consider $\VA(\LetTerm_1, \mathcal V\cap\FV{\LetTerm_1})\Def\LetG{\vec v'=e'}{\LetTerm'}$ and set:
		\[
			\VA(\LetTerm, \mathcal V)\Def
			\LetG{\Seq{\vec v_1,\vec v'}=\LetG{\vec v_1=e_1}{\Seq{\vec v_1,e'}}}{\LetTerm'}\,.
		\]

	\item If $\mathcal V\cap\FV{e_1}\neq\emptyset$ and $\vec v^a_1=f$. Notice that, by hypothesis, $f$ does not appear in the output of $\LetTerm_1$, as $\LetTerm$ (and hence $\LetTerm_1$) is positive. So we can consider $\VA(\LetTerm_1, (\mathcal V\cap\FV{\LetTerm_1})\cup\{f\})\Def\LetG{\vec v'=e'}{\LetTerm'}$ and define:
		\[
			\VA(\LetTerm, \mathcal V)\Def
			\LetG{
				\Seq{\vec v^+_1,\vec v'} =
				\LetG{\vec v_1 = e_1}{\Seq{\vec v_1^+, e'}}
			}
			{\LetTerm'}\,,
		\]
		\noindent if $\vec v^+_1$ is non-empty, otherwise: $\VA(\LetTerm, \mathcal V)\Def \LetG{\vec v'=\LetG{\vec v_1=e_1}{e'}}{\LetTerm'}$.
\end{enumerate}
}
\end{definition}

Finally, we can define the procedure $\VEL(\LetTerm, x)$. 
 This procedure basically consists in three steps: (i), it uses $\VA$ for gathering in a unique definition all expressions having a free occurrence of $x$ or a free occurrence of an arrow variable depending from $x$; then (ii), it performs $\Mult$ and $\Elimin$ rewriting so to make $x$ local to a definition, and finally (iii), it uses $\SD$ to move the obtained definition as the first definition of the let-term. This latter step is not strictly necessary but it is convenient in order to avoid free arrow variables of the expression having $x$ local, so getting a simple representation of the factor obtained after $x$ ``elimination''\todom{one can probably eliminate this last step}.

\begin{definition}[Variable elimination strategy]\label{def:VE}
The let-term $\VEL(\LetTerm, x)$ is defined from a positive let-term $\LetTerm \Def \Let{\vec v_1}{e_1}{\LetTerm_1}$ and a positive variable $x$ defined in $\LetTerm$ but not in the output of $\LetTerm$. The definition is by induction on $\LetTerm$ and splits in the following cases.
\SLV{
\begin{enumerate}
\item If $x\in\FV{\vec v_1}$ and $x\notin\FV{\LetTerm_1}$, then write by $\vec v_1'$ the pattern obtained from $\vec v_1$ by removing $x$ and define:
	$
		\VEL(\LetTerm, x) \Def \LetG{\vec v_1'=\LetG{\vec v_1=e_1}{\vec v_1'}}{\LetTerm_1}$.
		
\item If $x\in\FV{\vec v_1}$ and $x\in\FV{\LetTerm_1}$, then write by $\vec v_1'$ the pattern obtained from $\vec v_1$ by removing $x$. Remark that  $\LetTerm_1$ has at most one definition, as $x$ is not in the output of $\LetTerm_1$.  We split in two subcases:

	\begin{enumerate}
	\item if $\vec v_1'$ is positive, then set  $\LetG{\vec v'=e'}{\LetTerm'}\Def\VA(\LetTerm_1,\{x\})$ and define:\\
	$
		\VEL(\LetTerm, x) \Def \LetG{\Seq{\vec v_1',\vec v'}=\LetG{\vec v_1=e_1}{\Seq{\vec v_1', e'}}}{\LetTerm'}$.
	
	\item if $(\vec v'_1)^a=f$, then set $\LetG{\vec v'=e'}{\LetTerm'}\Def\VA(\LetTerm_1,\{x,f\})$ and define:\\
	$
		\VEL(\LetTerm, x) \Def \LetG{\Seq{\vec v_1^+,\vec v'}=\LetG{\vec v_1=e_1}{\Seq{\vec v_1^+, e'}}}{\LetTerm'}$.

	\end{enumerate}
	In both sub-cases, if $\vec v_1'^+$ is empty, we mean $\VEL(\LetTerm, x) \Def \LetG{\vec v'=\LetG{\vec v_1=e_1}{e'}}{\LetTerm'}$.
\item  If $x\notin\FV{ \vec v_1}$, then $x$ is defined in $\LetTerm_1$, and we can set:\\
	$
		\VEL(\LetTerm, x) \Def \SD(\LetG{\vec v_1=e_1}{\VEL(\LetTerm_1, x))}$.
\end{enumerate}
}{
\begin{enumerate}
\item If $x\in\FV{\vec v_1}$ and $x\notin\FV{\LetTerm_1}$, then write by $\vec v_1'$ the pattern obtained from $\vec v_1$ by removing $x$ and define:
	\[
		\VEL(\LetTerm, x) \Def \LetG{\vec v_1'=\LetG{\vec v_1=e_1}{\vec v_1'}}{\LetTerm_1}\,.
	\] 
\item If $x\in\FV{\vec v_1}$ and $x\in\FV{\LetTerm_1}$, then write by $\vec v_1'$ the pattern obtained from $\vec v_1$ by removing $x$. Remark that  $\LetTerm_1$ has at most one definition, as $x$ is not in the output of $\LetTerm_1$.  We split in two subcases:

	\begin{enumerate}
	\item if $\vec v_1'$ is positive, then set  $\LetG{\vec v'=e'}{\LetTerm'}\Def\VA(\LetTerm_1,\{x\})$ and define:
	\[
		\VEL(\LetTerm, x) \Def \LetG{\Seq{\vec v_1',\vec v'}=\LetG{\vec v_1=e_1}{\Seq{\vec v_1', e'}}}{\LetTerm'}\,.
	\] 
	
	\item if $(\vec v'_1)^a=f$, then set $\LetG{\vec v'=e'}{\LetTerm'}\Def\VA(\LetTerm_1,\{x,f\})$ and define:
	\[
		\VEL(\LetTerm, x) \Def \LetG{\Seq{\vec v_1^+,\vec v'}=\LetG{\vec v_1=e_1}{\Seq{\vec v_1^+, e'}}}{\LetTerm'}\,.
	\] 
	\end{enumerate}
	In both sub-cases, if $\vec v_1'^+$ is empty, we mean $\VEL(\LetTerm, x) \Def \LetG{\vec v'=\LetG{\vec v_1=e_1}{e'}}{\LetTerm'}$.
\item  If $x\notin\FV{ \vec v_1}$, then $x$ is defined in $\LetTerm_1$, and we can set:
	\[
		\VEL(\LetTerm, x) \Def \SD(\LetG{\vec v_1=e_1}{\VEL(\LetTerm_1, x))}\,.
	\] 
\end{enumerate}
}

As for $\VEF$, we extend $\VEL$ to sequences of (positive) variables, by 
\[
	\VEL(\LetTerm, (x_1,\dots, x_h))\Def\VEL(\VEL(\LetTerm, x_1), (x_2,\dots, x_h))\,.
\]
with the identity on $\LetTerm$ for the empty sequence. 
\end{definition}

\afterpage{%
\begin{landscape}
\begin{figure}
\begin{align*}
\LetTerm=\;&
\LetG{
	{\color{blue}x_1 = \StocA_1;
	x_2 = \StocA_2x_1; }
	x_3=\StocA_3x_2; 
	x_4=\StocA_4; 
	x_5=\StocA_5\Seq{x_3, x_4}; 
	x_6=\StocA_6(x_2, x_5)
}{
	\Seq{x_3,x_6}
}
&{\scriptstyle\text{(L1)}}
\\
\xrightarrow{\mathmakebox[0.4cm]{\Mult}}\;&
\LetG{
	{\color{blue}\Seq{x_1, x_2} = \LetG{x_1\!=\!\StocA_1}{\Seq{x_1,\StocA_2x_1}}; }
	x_3=\StocA_3x_2; 
	x_4=\StocA_4; 
	x_5=\StocA_5\Seq{x_3, x_4}; 
	x_6=\StocA_6(x_2, x_5)
}{
	\Seq{x_3,x_6}
}
&{\scriptstyle\text{(L2)}}
\\
\xrightarrow{\mathmakebox[0.4cm]{\Elimin_{x_1}}}\;&
\LetG{
	x_2=
		\underbrace{\LetG{
			\Seq{x_1, x_2} = \LetG{x_1\!=\!\StocA_1}{\Seq{x_1,\StocA_2x_1}}
			}{
			x_2
			}}_{e_1}
	;
	x_3=\StocA_3x_2; 
	x_4=\StocA_4;
	{\color{blue}x_5=\StocA_5\Seq{x_3, x_4};
	x_6=\StocA_6(x_2, x_5)}
}{
	\Seq{x_3,x_6}
}
&{\scriptstyle\text{(L3)}}
\\
\xrightarrow{\mathmakebox[0.4cm]{\Swap_2}}\;&
\LetG{
	x_2=e_1;
	x_3=\StocA_3x_2; 
	{\color{blue}x_4=\StocA_4;
	f=\lambda y. \StocA_6(x_2, y);}
	x_5=\StocA_5\Seq{x_3, x_4};
	x_6=f x_5;
	}{\Seq{x_3,x_6}}
&{\scriptstyle\text{(L4)}}
\\
\xrightarrow{\mathmakebox[0.4cm]{\Swap_1}}\;&
\LetG{
	x_2=e_1;
	{\color{blue}x_3=\StocA_3x_2; 
	f=\lambda y. \StocA_6(x_2, y);}
	x_4=\StocA_4;
	x_5=\StocA_5\Seq{x_3, x_4};
	x_6=f x_5;
	}{\Seq{x_3,x_6}}
&{\scriptstyle\text{(L5)}}
\\
\xrightarrow{\mathmakebox[0.4cm]{\Mult}}\;&
\LetG{
	{\color{blue}x_2=e_1;
	\Seq{x_3, f}=
		\LetG{
			x_3=\StocA_3x_2}
		{
			\Seq{x_3, \lambda y. \StocA_6(x_2, y)}
		};}
	x_4=\StocA_4;
	x_5=\StocA_5\Seq{x_3, x_4};
	x_6=f x_5;
	}{\Seq{x_3,x_6}}
&{\scriptstyle\text{(L6)}}
\\
\xrightarrow{\mathmakebox[0.4cm]{\Mult}}\;&
\LetG{
	{\color{blue}\Seq{x_2,\Seq{x_3, f}}=
		\LetG{
			x_2=e_1
		}{
			\Seq{
				x_2,
				\LetG{
					x_3=\StocA_3x_2
				}{
					\Seq{x_3, \lambda y. \StocA_6(x_2, y)}
				}
			}
		}
	;}
	x_4=\StocA_4;
	x_5=\StocA_5\Seq{x_3, x_4};
	x_6=f x_5;
	}{\Seq{x_3,x_6}}
&{\scriptstyle\text{(L7)}}
\\
\xrightarrow{\mathmakebox[0.4cm]{\Elimin_{x_2}}}\;&
\LetG{
	\Seq{x_3, f}=\underbrace{
		\LetG{
			\Seq{x_2,\Seq{x_3, f}}=
				\LetG{
					x_2=e_1
				}{
					\Seq{
						x_2,
						\LetG{
							x_3=\StocA_3x_2
						}{
							\Seq{x_3, \lambda y. \StocA_6(x_2, y)}
						}
					}
				}
		}{\Seq{x_3,f}}
		}_{e_{2}}
	;
	{\color{blue}x_4=\StocA_4;
	x_5=\StocA_5\Seq{x_3, x_4};}
	x_6=f x_5;
	}{\Seq{x_3,x_6}}
&{\scriptstyle\text{(L8)}}
\\
\xrightarrow{\mathmakebox[0.4cm]{\Mult}}\;&
\LetG{
	\Seq{x_3, f}=e_2;
	{\color{blue}\Seq{x_4,x_5}=
		\LetG{
			x_4=\StocA_4
		}{
			\Seq{x_4, \StocA_5\Seq{x_3, x_4}}	
		}
	;}
	x_6=f x_5;
	}{\Seq{x_3,x_6}}
&{\scriptstyle\text{(L9)}}
\\
\xrightarrow{\mathmakebox[0.4cm]{\Elimin_{x_4}}}\;&
\LetG{
	{\color{blue}\Seq{x_3, f}=e_2;
	x_5=
		\LetG{
			\Seq{x_4,x_5}=
				\LetG{
					x_4=\StocA_4
				}{
					\Seq{x_4, \StocA_5\Seq{x_3, x_4}}	
				}
		}{
				x_5
		}
	;}
	x_6=f x_5;
	}{\Seq{x_3,x_6}}
&{\scriptstyle\text{(L10)}}
\\
\xrightarrow{\mathmakebox[0.4cm]{\Swap_2}}\;&
\LetG{
	g = 	
		\underbrace{
		\lambda z.
		\LetG{
				\Seq{x_4,x_5}=
					\LetG{
						x_4=\StocA_4
					}{
						\Seq{x_4, \StocA_5\Seq{z, x_4}}	
					}
			}{
					x_5
			}
			}_{e_4}
	;
	\Seq{x_3, f}=e_2;
	{\color{blue}x_5= g x_3;
	x_6=f x_5;}
	}{\Seq{x_3,x_6}}
&{\scriptstyle\text{(L11)}}
\\
\xrightarrow{\mathmakebox[0.4cm]{\Mult}}\;&
\LetG{
	g = e_4;
	\Seq{x_3, f}=e_2;
	\Seq{x_5,x_6}= 
		\LetG{
			x_5 = g x_3
		}{
			\Seq{x_5, f x_5}
		}
	}{\Seq{x_3,x_6}}
&{\scriptstyle\text{(L12)}}
\\
\xrightarrow{\mathmakebox[0.4cm]{\Elimin_{x_5}}}\;&
\LetG{
	g = e_4;
	\Seq{x_3, f}=e_2;
	{\color{blue}x_6 =
		\LetG{
			\Seq{x_5,x_6}=	
				\LetG{
					x_5 = g x_3
				}{
					\Seq{x_5, f x_5}
				}
		}{
			x_6
		}}
	}{\Seq{x_3,x_6}}
&{\scriptstyle\text{(L13)}}
\\
\xrightarrow{\mathmakebox[0.4cm]{\Swap_3}}\;&
\LetG{
	{\color{blue}g = e_4;
	\Seq{x_3, x_6}=
		\LetG{
			\Seq{x_3, f}=e_2	
		}{
			\Seq{
				x_3,
				\LetG{
					\Seq{x_5,x_6}=	
						\LetG{
							x_5 = g x_3
						}{
							\Seq{x_5, f x_5}
						}
				}{
					x_6
				}		
			}
		}}
	}{\Seq{x_3,x_6}}
&{\scriptstyle\text{(L14)}}
\\
\xrightarrow{\mathmakebox[0.4cm]{\Swap_3}}\;&
\LetG{
	\Seq{x_3, x_6}=
		\underbrace{
		\LetG{
				g = e_4;		
		}{
			\LetG{
				\Seq{x_3, f}=e_2	
			}{
				\Seq{
					x_3,
					\LetG{
						\Seq{x_5,x_6}=	
							\LetG{
								x_5 = g x_3
							}{
								\Seq{x_5, f x_5}
							}
					}{
						x_6
					}		
				}
			}	
		}
		}_{e_5}
	}{\Seq{x_3,x_6}}
&{\scriptstyle\text{(L15)}}
\end{align*}
\caption{Rewriting of $\LetTerm$ into $\VEL(\LetTerm, (x_1, x_2, x_4, x_5))=\LetTerm'$ for $\LetTerm,\LetTerm'$ given in Ex.~\ref{ex:bayesian_graph_and_fact_as_let-term}. We underline in {\color{blue}blue} the fired redex in the following reduction step. We also name $e_1$, $e_2$, $e_4$, $e_5$, the expressions keeping local the corresponding variable (\ie $e_i$ keeps local $x_i$).}\label{fig:example_reduction}
\end{figure}
\end{landscape}
}

\begin{example}\label{ex:VEx_2}
Consider Fig.~\ref{fig:example_reduction} and denote by $\LetTerm_i$ the let-term in line $(Li)$. This figure details the rewriting sequence of the term $\LetTerm_1$ into $\LetTerm_{15}=\VEL(\LetTerm_1, (x_1, x_2, x_4, x_5))$. Namely, $\LetTerm_3=\VEL(\LetTerm_1, x_1)$, $\LetTerm_8=\VEL(\LetTerm_3, x_2)$, $\LetTerm_{11}=\VEL(\LetTerm_8, x_4)$, $\LetTerm_{15}=\VEL(\LetTerm_{11}, x_5)$.
\end{example}

\begin{proposition}[Rewriting into $\VEL$]\label{prop:rewriting_VE_well-typed}
Let $\LetTerm$ be a let-term with $n$ definitions: $\VEL(\LetTerm, x)$ is obtained from $\LetTerm$ by at most $n$ steps of the $\Red$ rewriting of Fig.~\ref{fig:let-rewriting}. 
In particular, $\VEL(\LetTerm, x)$ has the same type and free variables of $\LetTerm$.
\end{proposition}

The following theorem states both the soundness and completeness of our syntactic definition of $\VEL$ with respect to the more standard version defined on factors. The soundness is because any syntactic elimination variable is equivalent to the semantic $\VEL$ modulo the map $\Facts{\LetTerm}$. Completeness is because this holds for \emph{any} chosen variable, so all variable elimination sequences can be simulated in the syntax (Corollary~\ref{cor:soundnessVE_Let_nary}). \SLV{The proofs are in Appendix \ref{proof:let-rewriting}.}{}
\begin{theorem}\label{th:soundnessVE_Let}
Given $\LetTerm$ and $x$ as in Def.~\ref{def:VE}, we have:
$
	\Facts{\VEL(\LetTerm, x)} = \VEF(\Facts{\LetTerm}, x)$.
\end{theorem}

From Theorem~\ref{th:soundnessVE_Let} and Def.~\ref{def:VE_factors} and~\ref{def:VE},  the following is immediate.
\begin{corollary}\label{cor:soundnessVE_Let_nary}
Given a let-term $\LetTerm$ with all output variables positive and given a sequence $(x_1,\dots,x_n)$ of positive variables defined in $\LetTerm$ and not appearing in the output of  $\LetTerm$, we have that:
$
	\VEF(\Facts{\LetTerm},(x_1,\dots, x_n)) = \Facts{\VEL(\LetTerm, (x_1,\dots, x_n))}
$.
\end{corollary}

Recall from Ex.~\ref{ex:bayesian_graph_and_fact_as_let-term} that Bayesian networks can be represented by let-terms, so the above result shows that $\VEL$ implements in $\LetCalculus$ the elimination of a set of random variables of a Bayesian network in any possible order. 
It is well-known that the  variable elimination algorithm may  produce intermediate factors that are not stochastic matrices. The standard literature on probabilistic graphical models refer to the intermediate factors  simply as vectors of non-negative real numbers, missing any  finer characterisation. We stress that our setting allows for a more precise characterisation of such factors, as they are represented by \emph{well-typed} terms of $\LetCalculus$: not all non-negative real numbers vectors fit in. In particular, the typing system suggests a hierarchy of the complexity of a factor that, by recalling Remark~\ref{rk:ll_fragment}, can by summarised by the alternation between direct sums $\oplus$ and products $\&$: the simplest factors have type  $\bigoplus_n1$, i.e.~probabilistic distributions over $n$ values, then we have those of type $\bigwith_m\bigoplus_n1$, i.e.~stochastic matrices describing probabilities over $n$ values conditioned from observations over $m$ values, then we have  more complex factors of type $\bigoplus_k\bigwith_m\bigoplus_n1$, i.e.~probabilistic distributions over~stochastic matrices, and so forth.

\subsection{Complexity Analysis}\label{subsect:complexity}

Prop.~\ref{prop:rewriting_VE_well-typed} gives a bound to the number of $\Red$ steps needed to rewrite $\LetTerm$ into $\VEL(\LetTerm, x)$, however some of these steps adds new definitions in the rewritten let-term. The size of $\VEL(\LetTerm, x)$, although greater in general than that of $\LetTerm$, stays reasonable, in fact it has un upper bound linear in the degree of $\Facts{\LetTerm}_x$ (Prop.~\ref{prop:size_VE}). 
We define the size of an expression as follows:
\begin{align*}
\Size(v)&\Def 1
&\Size(\Abs{\vec v}{e})&\Def \Size(\vec v) + \Size(e)
&\Size(\Seq{e,e'})&\Def \Size(e)+\Size(e')
\\
\Size(\App f{\vec x})&\Def1+\Size(\vec x)
&\Size(\StocA(\vec x))&\Def 1+\Size(\vec x)
&\Size(\LetG{\vec v = e}{e'})&\Def \Size(\vec v) + \Size(e) + \Size(e')
\end{align*}
%
By induction on  $\LetTerm$, we obtain the following (see Appendix~\ref{proof:let-rewriting}):
\begin{proposition}\label{prop:size_VE}
Given a let-term $\LetTerm$ and a positive variable $x$ as in Def.~\ref{def:VE}, we have that 
$\Size(\VEL(\LetTerm, x))\leq \Size(\LetTerm) + 4 \times\Size(\FactVar{\Facts{\LetTerm}_x}\setminus\FV{\LetTerm})$.
\end{proposition}

\section{Conclusions and discussion}\label{sec:conclusions}

We have identified a fragment $\LetCalculus$ of the linear simply-typed $\lambda$-calculus which can express syntactically \emph{any} factorisation induced by a run of the variable elimination algorithm over a Bayesian network. In particular, we  define a rewriting (Fig.~\ref{fig:let-rewriting}) and a reduction strategy $\VEL$ (Def.~\ref{def:VE}) that, given a sequence $(x_1,\dots, x_n)$ of variables to eliminate,  transforms in $O(n\Size(\LetTerm))$ steps a let-term $\LetTerm$ into a let-term $\VEL(\LetTerm, (x_1,\dots, x_n))$ associated with the factorisation generated by the $(x_1,\dots, x_n)$ elimination (Corollary~\ref{cor:soundnessVE_Let_nary}).  We have proven that the size of $\VEL(\LetTerm, (x_1,\dots, x_n))$ is linear in the size of $\LetTerm$ and in the number of variables involved in the elimination process (Prop.~\ref{prop:size_VE}).

Our language is a fragment of a more expressive one \cite{EhrhardT19}, in which several classes of stochastic models can be encoded.  Our work is therefore a  step towards defining standard  exact  inference algorithms on a general-purpose stochastic language, as first propounded in \cite{KollerMP97} with the  
 goal is to have general-purpose algorithms of \emph{reasonable} cost, usable on any model expressed in the language.

While it is known (see \cite{KollerBook}, Sect. 9.3.1.3) that $\VE$ produces intermediate factors that are not  conditional probabilities---\ie not stochastic matrices---  our approach is able to associate \emph{a term and a type} to such factors.
In fact, the types of the calculus $\LetCalculus$ give a logical description of the interdependences between the factors generated by the \VE{} algorithm: the  grammar is more expressive than just the types of stochastic matrices between tuples of booleans (Remark~\ref{rk:ll_fragment}). 

\paragraph{Discussion and perspectives.}
Since our approach is theoretical, the main goal has been to give a formal framework for proving \emph{the soundness and the completeness} of $\VEL$. For that sake, the rewriting rules of Fig.~\ref{fig:let-rewriting} are reduced to a minimum, in order to keep reasonable the number of  cases in the proofs. The drawback is that the rewritten terms have a lot of bureaucratic code, as the reader may realize by looking at Fig.~\ref{fig:example_reduction}.  Although this fact is not crucial from the point of view of the asymptotic complexity, when aiming at   a prototypical implementation, one may  enrich the rewriting system with more rules to avoid useless code. 

The grammar of let-terms recalls the notion of administrative normal form (abbreviated ANF), which is often used as an intermediate representation in the compilation of functional programs. In particular, let-terms and ANF share in common the restriction of applications to variables, so suggesting a precise evaluation order. Several optimisations are defined as transformations over ANF, even considering some \emph{let-floating} rules analogous to the ones considered in Fig.~\ref{fig:let-rewriting}, see e.g.~\cite{let-floating}. 
Comparing these optimisations is not trivial as the cost model is different. E.g.~\cite{let-floating} aims to reduce heap allocations, while here we are factoring algebraic expressions to minimise  floating-point operations.
 We plan to investigate more in detail  the possible interplay/interference between these techniques.


The quest for optimal factorisations 
is central not only  to Bayesian programming. In particular, these techniques can be applied to large fragments of $\lambda$-calculus, suggesting heuristics for making tractable the computation of the quantitative semantics 
of other classes of $\lambda$-terms than the one identified by $\LetCalculus$. This is of great interest in particular because these semantics are relevant in describing quantitative observational equivalences, as hinted for example by the full-abstraction results achieved in probabilistic programming, \eg~\cite{EhrPagTas14,EhahrdPT18fa,clairambault:hal-01886956}. 

Finally, while we have  stressed that our work is theoretical, we do not mean to say that foundational understanding in general, and this work in particular, is irrelevant to the practice.  Let us mention one such  perspective.  
 Factored inference is central to inference in graphical models, but
	scaling it up to the more complex  problems
	expressible as probabilistic programs proves difficult---research in this direction is only  beginning, and is mainly guided by implementation techniques  \cite{PfefferRKO18,MansinghkaSHRCR18,HoltzenBM20}. 
	We believe  that a foundational understanding of   factorisation  on the structure of the program---starting from the most elementary algorithms, as we do here--- is  also an important   step  to allow  progress in this direction. 

\paragraph{On dealing with evidence.}
\newcommand{\Rain}{\mathsf{Rain}}
\newcommand{\Wet}{\mathsf{Wet}}
\newcommand{\true}{\mathsf{true}}
We have focused  on the computation of  marginals, without explicitly  treating \emph{posteriors}. Our approach  could  easily be  adapted to deal with evidence (hence, posteriors), by extending  syntax and rewriting rules to include an \texttt{observe} construct as in  \cite{HoltzenBM20} or  in  \cite{FaggianPV24}. 
\SLV{}{We have preferred not to do this here in order to keep  presentation and proofs simple.     We stress however that  the main ingredients to  perform  the Variable Elimination algorithm remain the same when dealing with evidence.  Indeed, textbooks usually choose to  present exact inference  for  marginal distribution (typically via the \VE{} algorithm)---adapting to evidence is than straightforward. \todoc{We can stop here, or at the end of the paragraph}
   For the  details we  refer to  \cite{KollerBook} (9.3.2 “Dealing with Evidence”, page 303) and  \cite{Darwiche2009}, Ch. 6.7 (“Computing posterior marginals”).  
}


%
%

\begin{credits}
\paragraph{Acknowledgements.}  
We are deeply grateful to Marco Gaboardi for suggesting investigating the link between variable elimination and linear logic, as well as to Robin Lemaire, with whom we initiated this research. 
Work supported by the ANR  grant PPS  ANR-19-CE48-0014
 and ENS de Lyon research grant.  
\end{credits}

 \bibliographystyle{splncs04}
 \bibliography{biblio}

\newpage
\appendix
\section*{Appendix}

\renewcommand{\RED}[1]{#1}
This Appendix includes more technical material, and missing proofs.

\section{Section~\ref{sect:syntax}}
\label{proof:pcoh}


\subsubsection{Probabilistic Coherence Spaces.} The model of weighted relational semantics is simple but it misses an essential information: the difference between the tensor or the arrow between two types\footnote{In categorical terms, the weighted relational semantics forms a category which is compact closed.}. This means that from a denotation $\Sem e\in\RP^{\Web{\PosA}\times\Web{\TypeA}}$ of a closed expression $e$, we do not know whether $e$ computes a pair of two expressions, one of type $\PosA$ and one of type $\TypeA$, or a function from type $\PosA$ to type $\TypeA$.  Probabilistic coherence spaces \cite{Girard2003,danosehrhard} can be seen as a kind of enrichment of $\RP$-weighted relational semantics which recover this information. 

We sketches here the denotational interpretation of  probabilistic coherence spaces and a nice consequence of it. This model associates an expression $e$ with the same matrix $\Sem e$ as in weighted relational semantics (Figure~\ref{fig:denotation_terms}), but it endows the web interpreting a type $\TypeA$ with a set $\Clique{\TypeA}\subseteq\RP^{\Web\TypeA}$ of so-called ``probabilistic cliques'', so that $\Sem e$ can be proven to maps vectors in $\Clique{\FV e}$ to vectors in $\Clique{\Type e}$ (Proposition~\ref{prop:pcoh_correctness}). In particular, the denotation of a closed expression can be seen as a vector in $\Clique{\Type e}$. 

\bigskip
We only sketch here the model and we refer the reader to the literature, especially \cite{danosehrhard,EhahrdPT18fa}, for more details and the omitted proofs.

We define a \textdef{polar operation} on sets of vectors $\mathrm P\subseteq \RP^{\SetA}$  as
\begin{equation}\label{eq:polarity}
\mathrm P^\perp\Def\left\{\MatB\in \RP^{\SetA} \st \forall \MatA\in \mathrm P\ \sum_{a\in\SetA}\MatA_a\MatB_a\leq 1\right\}.
\end{equation}
 Polar satisfies the following immediate properties: $\mathrm P\subseteq
 \mathrm P^{\perp\perp}$, if $\mathrm P\subseteq\mathrm Q$ then
 $\mathrm Q^\perp\subseteq \mathrm P^\perp$, and then $\mathrm P^\perp=\mathrm P^{\perp\perp\perp}$. 

A \textdef{probabilistic coherence space}, or PCS for short, is a pair
$(\SetA, P)$ where $\SetA$ is a
finite set called the \textdef{web} of the space and
$P$ is a subset of $\RP^{\Web {\Pcs{X}}}$ satisfying:
\begin{enumerate}
\item $P^{\perp\perp}=P$,
\item $\forall \ValA\in\SetA$, $\exists \ScalA>0$, $\forall \MatA\in P$, $\MatA_\ValA\leq \ScalA$,
\item $\forall \ValA\in\SetA$, $\exists\MatA\in P$, $\MatA_\ValA>0$.
\end{enumerate}

Condition (1) is central and assures $P$ to have the closure properties necessary to interprets probabilistic programs (namely, convexity and Scott continuity). Condition (2)  requires the projection  of $P$   in   any  direction   to   be  bounded,   while (3)  forces  $P$   to    cover   every
direction\footnote{The conditions (2) and (3) are introduced in \cite{danosehrhard} for keeping finite all the scalars involved in the case of infinite webs, yet they are not explicitly stated in the definition of a PCS in \cite{Girard2003}. We consider appropriate to keep them also in the finite dimensional case, as they assure that the set $P$ is the unit ball of the whole cone  $\RP^{\SetA}$ endowed with the suitable norm.}.

Let us lift the weighted relational denotation of a type $\TypeA$ to PCS, by defining a set $\Clique{\TypeA}\subseteq \RP^{\Web{\TypeA}}$ so that the pair $(\Web{\TypeA}, \Clique{\TypeA})$ is a PCS. The definition of  $\Clique{\TypeA}$ is by induction on $\TypeA$:
\begin{align*}
\Clique{\Bool}&\Def\{(\ScalA,\ScalB)\in \RP^{\Web{\Bool}}\st \ScalA+\ScalB\leq 1\}\,,
\\
\Clique{\PosA\otimes\TypeA}&\Def\{\MatA\otimes\MatB\in\RP^{\Web{\PosA\otimes\TypeA}}\st \MatA\in\Clique{\PosA}, \MatB\in\Clique{\TypeA}\}^{\perp\perp}\,,
\\
\Clique{\PosA\multimap\TypeA}&\Def\{\MatA\in\RP^{\Web{\PosA\multimap\TypeA}}\st \forall\VecA\in\Clique{\PosA}, \App{\MatA}{\VecA}\in\Clique{\TypeA}\}\,,
\end{align*}
where the tensor $\MatA\otimes\MatB$ of two vectors $\MatA\in\RP^{\SetA}$, $\MatB\in\RP^{\SetB}$ is given by: $(\MatA\otimes\MatB)_{(\ValA,\ValB)}\Def\MatA_{\ValA}\,\MatB_{\ValB}$, for any $(\ValA, \ValB)\in\SetA\times\SetB$. 

\begin{example}
Recall that the web of the denotation of any type can be seen as, basically, a set of tuples of booleans ``structured'' by parenthesis. However, the set $\Clique{\TypeA}$ is different depending on which $\TypeA$ we consider.

If $\TypeA$ is positive, then $\Clique{\TypeA}$ contains the vectors representing the subprobabilistic distributions of the tuples in $\Web\TypeA$. For example, $\Clique{\Bool\otimes(\Bool\otimes\Bool)}$ is the set of vectors $(\ScalA_{(b_1,(b_2, b_3))})_{b_i\in\{\True,\False\}}$ of total mass $\sum_{b_1,b_2,b_3\in\{\True,\False\}}\ScalA_{(b_1,(b_2, b_3))}\leq1$. 

If $\TypeA$ is an arrow type 
between two positive types, then $\Clique{\TypeA}$ contains the sub-stochastic matrices between these two positive types. For example, $\Clique{\Bool\multimap(\Bool\otimes\Bool)}$ is the set of matrices $(\ScalA_{b_1,(b_2, b_3)})_{b_i\in\{\True,\False\}}$ such that for all $b_1\in\{\True,\False\}$, $\sum_{b_2,b_3\in\{\True,\False\}}\ScalA_{b_1,(b_2, b_3)}\leq 1$. 

If  $\TypeA$ is a general arrow type 
then the situation can be subtler. %
For instance $\Clique{\Bool\otimes(\Bool\multimap\Bool)}$ is the set of matrices $(\ScalA_{b_1,(b_2, b_3)})_{b_i\in\{\True,\False\}}$ such that there are \(\lambda^c\in\Clique{\Bool\multimap\Bool}\) for \(c=\True,\False\) (that is \(\lambda^c_{b,\True}+\lambda^c_{b,\False}\leq 1\) for \(b,c=\True,\False\)) and \(\alpha\in\Clique{\Bool}\) such that \(\lambda_{b_1,(b_2,b_3)}=\alpha_{b_1}\lambda^{b_1}_{b_2,b_3}\).
%
\end{example}

Moreover, we associate a finite set of variables $\mathcal V$ with a set $\Clique{\mathcal V}\subseteq \RP^{\Web{\mathcal V}}$:
\begin{align*}
\Clique{\mathcal V}&\Def\{\bigotimes_{v\in\mathcal V}\MatA_v\st \MatA_v\in\Clique{\Type v}\}^{\perp\perp}\,,
\end{align*}
where the indexed tensor $\bigotimes_{v\in\mathcal V}\MatA_v$ of a family of vectors $\MatA_v\in\RP^{\Type v}$, for $v\in\mathcal V$, is given by: $\bigl(\bigotimes_{v\in\mathcal V}\MatA_v\bigr)_{\FValA}\Def\prod_{v\in\mathcal V}(\MatA_v)_{\FValA}$, for $\FValA\in\Web{\mathcal V}$. 

\begin{proposition}[\cite{Girard2003,danosehrhard}]\label{prop:pcoh_correctness}
If $e$ is a well-typed expression, then: for every $\MatA\in\Clique{\FV e}$, $\MApp{\Sem e}{\MatA}\in\Clique{\Type e}$. 
\end{proposition}

With any positive type \(P\) we associate its dimension
\(\Ptdim P\in\Nat\) by \(\Ptdim\Bool=2\) and %
\(\Ptdim{P\otimes Q}=\Ptdim P\Ptdim Q\). This means that \(\Ptdim P\)
is the cardinality of \(\Web P\). And with any type \(T\) we associate
its height \(\Theight T\in\Nat\), the definition is: %
\(\Theight P=1\), \(\Theight{P\multimap T}=\Ptdim P\times\Theight
T\) %
and \(\Theight{P\otimes T}=\Theight T\). Then the following property
is easy to prove:
\begin{lemma}
  For any type \(T\) and any \(x\in\Clique T\) one has %
  \(\sum_{a\in\Web T}x_a\leq\Theight T\).
\end{lemma}
It can be strengthen as follows. Define \emph{probabilistic spaces
  with totality} as triples \((S,P,\mathcal T)\) where \((S,P)\) is a
probabilistic coherence space and \(\mathcal T\subseteq P\) satisfies
\(\mathcal T={\mathcal T}^{\perp\perp}\) for the following notion of
orthogonality:
\({\mathcal T}^\perp=\{x'\in P^{\perp} \St \forall x\in\mathcal T\
\sum_{a\in S}x_ax'_a=1\}\). The elements of \(\mathcal T\) are the
\emph{total cliques} of the probabilistic coherence space with
totality. Types can be interpreted as such objects. In the
interpretation of \(\Bool\), the total cliques \(x\) are those such that
\(x_\True+x_\False=1\) (the true probability distributions), and in
the interpretation of \(\Bool\multimap\Bool\) the total cliques are
the stochastic matrices, that is the matrices
\((\lambda_{b_1,b_2})_{b_1,b_2\in\{\True,\False\}}\) such that for all $b_1\in\{\True,\False\}$, $\sum_{b_2\in\{\True,\False\}}\ScalA_{b_1,b_2}= 1$, etc.
\begin{proposition}
\label{prop:totality-mass_App}
  For any type \(T\) and any \(x\in\Clique T\) which is total one has %
  \(\sum_{a\in\Web T}x_a=\Theight T\). In particular, any expression $e$ of $\LetCalculus$ is total. 
\end{proposition}
%


\section{Section~\ref{sect:factors}}
\label{proof:factors}



\begin{lemma}\label{lemma:facts_variables}
Let $\LetTerm\Def\LetG{\vec v_1=e_1; \dots; \vec v_n=e_n}{\vec v_{n+1}}$ be a let-term, then  
\[
	\FactVar{\Facts\LetTerm} = 
		\FV{\LetTerm}
		\uplus
		\left(\FV[a]{\vec v_{n+1}}\setminus\FV{\LetTerm}\right)
		\uplus
		\left(\biguplus_{i=1}^n\FV[+]{\vec v_i}\right)
\]
\end{lemma}
\begin{proof}
By induction on $n$.
\end{proof}

\begin{lemma}[Splitting]\label{lemma:splitting_facts}
Let $\LetTerm$ be the let-term $\LetG{\vec v=e}{\LetTerm'}$ and consider a set of variables $\mathcal V\subseteq\FV{\LetTerm}$ such that $\mathcal V\cap\FV{\LetTerm'}=\emptyset$ (so that $\mathcal V\subseteq\FV{e}$). We have:
\begin{itemize}
\item if $\vec v^a = f$ and $f$ is not an output variable,
\begin{align*}
	\Facts{\LetTerm}_{\mathcal V}&=\{\FSum{f}{\Fact{\vec v=e}\FProd\Facts{\LetTerm'}_{f}}\},&
	\Facts{\LetTerm}_{\neg (\mathcal V\cup\{f\})}&=\Facts{\LetTerm'}_{\neg \{f\}};
\end{align*}
\item otherwise, 
\begin{align*}
	\Facts{\LetTerm}_{\mathcal V}&=\{\Fact{\vec v=e}\},&
	\Facts{\LetTerm}_{\neg \mathcal V}&=\Facts{\LetTerm'}.\\
\end{align*}
\end{itemize}
\end{lemma}
\begin{proof}
By the hypothesis $\mathcal V\subseteq\FV{e}\setminus\FV{\LetTerm'}$, the definition $\Facts{\LetTerm}_{\mathcal V}$ should use $\Fact{\vec v=e}$. If moreover $\vec v$ contains an arrow variable not in the outputs, then the factor, if any, having such a variable in its set of variables will be multiplied with $\Fact{\vec v=e}$.\Myqed
\end{proof}

\begin{proof}[Detailed proof of Proposition~\ref{prop:factors_as_semantics}]
By induction on $\LetTerm$.

If $\LetTerm$ is just the tuple $\vec v$, then $\mathcal F=\FV{\vec v}$ and $\FSum{\mathcal F\setminus(\FV\LetTerm\cup\FV{\vec v})}{\BigFProd\Facts{\LetTerm}}$ maps the empty sequence into $1$. Moreover, $\FValA\mid_{\FV{\LetTerm}\cap\FV{\vec v}}=\FValA$ and similarly for $\FValB$, so $\Sem{\LetTerm}_{\FValA,\FValB}$ reduces to the Kronecker delta $\delta_{\FValA, \FValB}$. 

If $\LetTerm$ is $\Let{\vec v'}{e'}{\LetTerm'}$, then by definition of the semantics:
\begin{align}
\Sem{\LetTerm}_{\FValA,\FValB} 
&=\sum_{\FValC\in\Web{\vec v'}}\Sem{e'}_{\FValA\mid_{\FV{e'}}, \FValC} \Sem{\LetTerm'}_{(\FValA\FPlus\FValC)\mid_{\FV{\LetTerm'}},\FValB}
\label{proof_eq:sem_from_fact}
\end{align}
Suppose $\FValA\mid_{\FV{\LetTerm}\cap\FV{\vec v}}\neq\FValB\mid_{\FV{\LetTerm}\cap\FV{\vec v}}$, with $\vec v$ the output variables of $\LetTerm$ (as well as of $\LetTerm'$). We claim that: $\FV{\LetTerm}\cap\FV{\vec v}\subseteq\FV{\LetTerm'}\cap\FV{\vec v}$. In fact:
\begin{align*}
\FV{\LetTerm}\cap\FV{\vec v}&=\left(\left(\FV{\LetTerm'}\setminus\FV{\vec v'}\right)\cap\FV{\vec v}\right)\cup\left(\FV{e'}\cap\FV{\vec v}\right)\\
&=\left(\FV{\LetTerm'}\setminus\FV{\vec v'}\right)\cap\FV{\vec v}\\
&\subseteq\FV{\LetTerm'}\cap\FV{\vec v}
\end{align*}
The passage from the first to the second line is because, if $w\in\FV{e'}\cap\FV{\vec v}$, then $w\notin \FV{\vec v'}$ (as $\FV{\vec v'}\cap\FV{e'}=\emptyset$), as well as $w\in\FV{\ell'}\cap\FV{\vec v}$ (as we can suppose, by renaming, the variables of $\vec v$ bounded in $\ell'$ to be disjoint from the free variables of $e'$). 

Therefore, $\FValA\mid_{\FV{\LetTerm}\cap\FV{\vec v}}\neq\FValB\mid_{\FV{\LetTerm}\cap\FV{\vec v}}$ implies $(\FValA\FPlus\FValB)\mid_{\FV{\LetTerm'}\cap\FV{\vec v}}\neq\FValB\mid_{\FV{\LetTerm'}\cap\FV{\vec v}}$. So, by inductive hypothesis, $\Sem{\LetTerm'}_{(\FValA,\FValC)\mid_{\FV{\LetTerm'}},\FValB}=0$ and we conclude $\Sem{\LetTerm}_{\FValA,\FValB} =0$.

Otherwise, we can suppose $\FValA\mid_{\FV{\LetTerm}\cap\FV{\vec v}}=\FValB\mid_{\FV{\LetTerm}\cap\FV{\vec v}}$ and call $\vec h$ the family $\FValA' \FPlus \FValC \FPlus \FValB'$ for
$\FValA' = \FValA\mid_{\mathcal F\setminus \FV{\vec v}}$, 
$\FValB' = \FValA\mid_{\mathcal F\setminus \FV{\LetTerm}}$ and
$\FValC = \FValA\mid_{\FV{\LetTerm}\cap\FV{\vec v}}=\FValB\mid_{\FV{\LetTerm}\cap\FV{\vec v}}$. Let us also define $\mathcal F'=\FactVar{\Facts{\LetTerm'}}$.

We split in three subcases.
\begin{itemize}
\item If $(\vec v')^a =  f$ and $\Facts{\LetTerm'}_f=\{\FactA\}$. Let $(\vec v')^+=\vec x$ and notice that $\mathcal F=\FV{\vec x}\cup\FV{e'}\cup(\mathcal F'\setminus\{f\})$. We can rewrite the right-hand side term in Equation~\eqref{proof_eq:sem_from_fact} as:
\begin{align}
&\label{eqSoundFacts:0}=
	\sum_{\FValC'\in\Web{\vec x}}\sum_{d''\in\Web{f}}
	\Sem{e'}_{\FValA\mid_{\FV{e'}}, (\FValC', d'')} \Sem{\LetTerm'}_{(\FValA,(\FValC', d''))\mid_{\FV{\LetTerm'}},\FValB}
\\
&\label{eqSoundFacts:1}=
	\FSum{\vec x, f}{
		\Fact{\vec v' = e'}
		\FProd
		\FSum{\mathcal F'\setminus(\FV{\LetTerm'}\cup\FV{\vec v})}{\BigFProd\Facts{\LetTerm'}}
	}(\vec h)
\\
&\label{eqSoundFacts:2}=
	\FSum{\vec x, f}{
		\Fact{\vec v' = e'}
		\FProd
		\FSum{\mathcal F'\setminus(\FV{\LetTerm'}\cup\FV{\vec v})}{\FactA\FProd\BigFProd\Facts{\LetTerm'}_{\neg f}}
	}(\vec h)
\\
&\label{eqSoundFacts:3}=
	\FSum{\vec x}{
		\FSum{\mathcal F'\setminus(\FV{\LetTerm'}\cup\FV{\vec v})}{
			\FSum{f}{
				\Fact{\vec v' = e'}
				\FProd
				\FactA
			}
		\FProd\BigFProd\Facts{\LetTerm'}_{\neg f}
		}
	}(\vec h)
\\
&\label{eqSoundFacts:4}=
	\FSum{\vec x}{
		\FSum{\mathcal F'\setminus(\FV{\LetTerm'}\cup\FV{\vec v})}{
			\BigFProd\Facts{\LetTerm}
		}
	}(\vec h)
\\
&\label{eqSoundFacts:5}=
		\FSum{\mathcal F\setminus(\FV{\LetTerm}\cup\FV{\vec v})}{
			\BigFProd\Facts{\LetTerm}
		}(\vec h)
\end{align}
The passage to line~\eqref{eqSoundFacts:1} applies the inductive hypothesis to $\LetTerm'$ and Definition~\ref{def:factor_of_def}. The other passages use Proposition~\ref{prop:factor_prod_ass_distr}, in particular:
the passage to line~\eqref{eqSoundFacts:2} applies the hypothesis $\Facts{\LetTerm'}_f=\{\FactA\}$ and commutativity and associativity of $\FProd$.
The passage to line~\eqref{eqSoundFacts:3} applies twice the distribution of $\FProd$ over $\sum_S$, since the domain of the distributed factor is disjoint wrt $S$. 
The passage to line~\eqref{eqSoundFacts:4} applies the definition of n-ary factor product, and finally the passage to line~\eqref{eqSoundFacts:5} applies the associativity and commutativity of $\sum_S$, remarking that $\FV{\vec x} \cup\mathcal F'\setminus(\FV{\LetTerm'}\cup\FV{\vec v})=\mathcal F\setminus(\FV{\LetTerm}\cup\FV{\vec v})$.

\item If $(\vec v')^a = f$  and $\Facts{\LetTerm'}_f=\emptyset$, so that $f$ is an output variable of $\LetTerm'$. Let $(\vec v')^+ = \vec x$, we have: $\mathcal F=\FV{\vec x}\cup\FV{e'}\cup\mathcal F'$.

Equation~\eqref{proof_eq:sem_from_fact}  can be rewritten in line~\eqref{eqSoundFacts:0} above. However, the variable $f$ is now in $\mathcal F$, as it is an output variable of $\LetTerm$, as well as of $\LetTerm'$. So, by inductive hypothesis, the terms of the sum over $\Web{f}$ are non-zero only for $d''=\FValB_f$, so we get (by an analogous reasoning as before):
\begin{align*}
&=
	\FSum{\vec x}{
		\Fact{\vec v' = e'}
		\FProd
		\FSum{\mathcal F'\setminus(\FV{\LetTerm'}\cup\FV{\vec v})}{\BigFProd\Facts{\LetTerm'}}
	}(\vec h)
\\
&=
	\FSum{\vec x}{
		\FSum{\mathcal F'\setminus(\FV{\LetTerm'}\cup\FV{\vec v})}{
			\Fact{\vec v' = e'}
			\FProd
			\BigFProd\Facts{\LetTerm'}
		}
	}(\vec h)
\\
&=
	\FSum{\vec x}{
		\FSum{\mathcal F'\setminus(\FV{\LetTerm'}\cup\FV{\vec v})}{
			\BigFProd\Facts{\LetTerm}
		}
	}(\vec h)
\\
&=
		\FSum{\mathcal F\setminus(\FV{\LetTerm}\cup\FV{\vec v})}{
			\BigFProd\Facts{\LetTerm}
		}(\vec h)
\end{align*}

\item The case $\vec v'$ has no arrow variable is completely similar to the previous one, just we do not need to consider the sum over a positive variable. 
\end{itemize}
\Myqed
\end{proof}

\begin{proof}[Detailed proof of Proposition~\ref{prop:order_sequence}]
By induction on the length $h$ of the sequence $\Seq{v_1,\dots, v_h}$. The induction step uses Proposition~\ref{prop:factor_prod_ass_distr}:
\begin{align*}
	\BigFProd\VEF(\FactSetA, \Seq{v_1,\dots, v_h})
	&=\BigFProd\VEF(\VEF(\FactSetA, v_1), \Seq{\LabelA_2,\dots, v_h})\\
	&=\BigFProd\VEF\left(\left(\FSum{\LabelA_1}{\BigFProd\FactSetA_{\LabelA_1}}\right)\FProd\BigFProd\FactSetA_{\neg\LabelA_1}, \Seq{\LabelA_2,\dots, \LabelA_h}\right)\\
	&=\BigFProd\VEF\left(\FSum{\LabelA_1}{\BigFProd\FactSetA_{\LabelA_1}\FProd\BigFProd\FactSetA_{\neg\LabelA_1}}, \Seq{\LabelA_2,\dots, \LabelA_h}\right)\\
	&=\BigFProd\VEF\left(\FSum{\LabelA_1}{\BigFProd\FactSetA}, \Seq{\LabelA_2,\dots, \LabelA_h}\right)\\
	&=\FSum{\{\LabelA_2,\dots, \LabelA_h\}}{\FSum{\LabelA_1}{\BigFProd\FactSetA}}\\
	&=\FSum{\{\LabelA_1,\dots, \LabelA_h\}}{\BigFProd\FactSetA}
\end{align*}
The equality $\FactVar{\VEF(\FactSetA, \Seq{\LabelA_1,\dots, \LabelA_h})} = \FactVar{\FactSetA}\setminus\{\LabelA_1,\dots, \LabelA_h\}$ follows because by definition $\FactVar{\BigFProd\FactSetA}=\FactVar{\FactSetA}$.
\Myqed
\end{proof}


\RED{
	A consequence of Prop.~\ref{prop:order_sequence} is that the factor $\BigFProd\VEF(\FactSetA, \Seq{\LabelA_1,\dots, \LabelA_h})$ is independent from the order of the variables appearing in the sequence, which means  
	$\BigFProd\VEF(\FactSetA, \Seq{\LabelA_1,\dots, \LabelA_h})=\BigFProd\VEF(\FactSetA, \Seq{\LabelA_{\sigma(1)},\dots, \LabelA_{\sigma(h)}})$ for any permutation $\sigma$. However, $\VEF(\FactSetA, \Seq{\LabelA_1,\dots, \LabelA_h})$ and $\VEF(\FactSetA, \Seq{\LabelA_{\sigma(1)},\dots, \LabelA_{\sigma(h)}})$ are in general different sets of factors that can compute the product with different performances. 
	
	
	\begin{proposition}
		\label{prop:VEfact_complexity}
		Given $\FactSetA$ and $\Seq{\LabelA_1,\dots, \LabelA_h}$ as in Def.~\ref{def:VE_factors}, let $d$ be the maximal degree $\max\{\Deg_{\VEF(\FactSetA, \Seq{\LabelA_1,\dots, \LabelA_i})_{\LabelA_{i+1}}} \st 0\leq i< h\}$.
		Then, the set $\VEF(\FactSetA, \Seq{\LabelA_1,\dots, \LabelA_h})$ can be computed out of $\FactSetA$ in $O\left(h(\Base_{\FactSetA})^{d}\right)$ basic operations.
		
		Moreover, if $d'=\max(d, \Deg_{\VEF(\FactSetA, \Seq{\LabelA_1,\dots, \LabelA_h})})$, then the factor $\FSum{\{\LabelA_1,\dots, \LabelA_h\}}{\BigFProd\FactSetA}$ can be computed out of $\FactSetA$ in $O\left(h(\Base_{\FactSetA})^{d'}\right)$ basic operations.
	\end{proposition}

\begin{proof}
	Definition~\ref{def:VE_factors} hints a computation giving $\VEF(\FactSetA, \Seq{\LabelA_1,\dots, \LabelA_h})$ by computing $\FactSetA^i\Def\VEF(\FactSetA, \Seq{\LabelA_1,\dots, \LabelA_i})$, for $0\leq i\leq h$. In fact, $\FactSetA^0=\FactSetA$ and, for $i>0$,
	\[
	\FactSetA^i = \{\FSum{\LabelA_i}{\BigFProd{\FactSetA^{i-1}}_{\LabelA_i}}\}\uplus{\FactSetA^{i-1}}_{\neg\LabelA_i}
	\]
	The computation of the new factor $\FSum{\LabelA_i}{\BigFProd{\FactSetA^{i-1}}_{\LabelA_i}}$ requires 
	$O({
		\Base_{
			{\FactSetA^{i-1}}_{\LabelA_{i}}
		}
	}^{\Deg_{{\FactSetA^{i-1}}_{\LabelA_{i}}}})$ which is bounded by $O({
		\Base_{\FactSetA}
	}^{\Deg_{{\FactSetA^{i-1}}_{\LabelA_{i}}}})$, as $\FactVar{\FactSetA^{i-1}}\subseteq\FactVar{\FactSetA}$. If we suppose that spitting $\FactSetA^{i-1}$ into ${\FactSetA^{i-1}}_{\LabelA_{i}}$ and ${\FactSetA^{i-1}}_{\neg\LabelA_{i}}$ requires a negligible number of operations, as the number of factors in $\FactSetA^{i-1}$ is bounded by $\Deg_{\FactSetA^{i-1}}$, then the whole computation of $\FactSetA^i$ out of $\FactSetA^{i-1}$ requires $O({
		\Base_{\FactSetA}
	}^{\Deg_{{\FactSetA^{i-1}}_{\LabelA_{i}}}})$ basic operations. Since the computations of the various $\FactSetA^{i}$ are sequential, we get the bound $h(\Base_{\FactSetA})^{d}$ by taking the maximal degree $d$ of such a $\FactSetA^{i}$.\Myqed
\end{proof}
}

\section{Section ~\ref{sect.VE_as_let_rewriting}}
\label{proof:let-rewriting}

\RED{\begin{proof}[Proof sketch of Prop.~\ref{prop:let_rewriting_subject}  ]
		By structural induction on $\LetTerm$, splitting on the various cases of $\Red$. The subtle case is for  $\Swap_3$, where we remark that if $\vec v_1^a=f$, then the linear typing system assures that: $f\in\FV{e_2}$ iff $f\notin\FV{\LetTerm}$. 
		\Myqed
\end{proof}}

\RED{\begin{proof}[Proof sketch of Prop.~\ref{prop:let-rewriting_semantics}]
		A direct proof by induction can be quite cumbersome to develop in full details. A simpler way to convince about this statement is by noticing that $\LetTerm\Red\LetTerm'$ implies that the two let-terms are $\beta$-equivalent if one translates $\Let{\vec v}{e}{e'}$ into $(\lambda \vec v.e')e$. This translation preserves the semantics and weighted relational semantics is known to be invariant under $\beta$-reduction (see e.g.~\cite{LairdMMP13}).\Myqed
\end{proof}}

\begin{proof}[Detailed proof of Lemma~\ref{lemma:facts-invariant}]
By inspection of cases. Concerning  $\Swap_2$,
notice that whenever an arrow variable $f$ defined by some $\vec v_i$ does not appear in the output of $\LetTerm$, then by Lemma~\ref{lemma:facts_variables}, $f\notin\FactVar{\Facts{\LetTerm}}$ so that the arrow variable created in the contractum of $\Swap_2$ is 
not in the set of variables of the factors associated with them. \Myqed
\end{proof}

\begin{proof}[Detailed proof of Lemma~\ref{lemma:VA_semantic}]
If $\mathcal V=\emptyset$, then $\VA(\LetTerm, \mathcal V)=\LetTerm$. Otherwise, we prove $\Facts{\VA(\LetTerm, \mathcal V)} = 	\left\{\BigFProd\Facts{\LetTerm}_{\mathcal V}\right\}\uplus\Facts{\LetTerm}_{\neg \mathcal V}$ by induction on $\LetTerm$, splitting according to the cases of Definition~\ref{definition:VA} from which we adopt the notation. Case 1 is trivial.

\begin{itemize}
\item Case 2 of Def.~\ref{definition:VA}: by Lemma~\ref{lemma:SD_soundness}, $\Facts{\VA(\LetTerm, \mathcal V)}=\Facts{\LetG{\vec v_1=e_1}{\VA(\LetTerm_1, \mathcal V)}}$. 
We split in three subcases.


	\begin{itemize}
		\item If $\vec v_1$ is positive, then by applying the inductive hypothesis we get that $\Facts{\LetG{\vec v_1=e_1}{\VA(\LetTerm_1, \mathcal V)}}$  is equal to:
	\begin{equation*}
	\left\{
		\BigFProd(\Facts{\LetTerm_1}_{\mathcal V}		
	\right\}
	\uplus
	\left\{
		\Fact{\vec v_1=e_1}
	\right\}
	\uplus
	\Facts{\LetTerm_1}_{\neg \mathcal V}
	=
	\left\{
		\BigFProd(\Facts{\LetTerm}_{\mathcal V}		
	\right\}
	\uplus
	\Facts{\LetTerm}_{\mathcal V}.
	\end{equation*}

	\item If $\vec v_1$ has an arrow variable $f$ then this variable cannot be in the output, so, by linear typing, either $f\in\FactVar{\Facts{\LetTerm_1}_{\mathcal V}}\setminus\FactVar{\Facts{\LetTerm_1}_{\neg \mathcal V}}$ or $f\in\FactVar{\Facts{\LetTerm_1}_{\neg \mathcal V}}\setminus\FactVar{\Facts{\LetTerm_1}_{\mathcal V}}$. In the first case, by applying the inductive hypothesis we get that $\Facts{\LetG{\vec v_1=e_1}{\VA(\LetTerm_1, \mathcal V)}}$ is equal to: 	
	\begin{align*}
		&\{
			\FSum{f}{
				\Fact{\vec v_1=e_1}
				\FProd
				\BigFProd\Facts{\LetTerm_1}_{\mathcal V}		
			}
		\}
		\uplus
		\Facts{\LetTerm_1}_{\neg (\mathcal V\cup\{f\})}
		\\
		&=
		\{
			\BigFProd\Facts{\LetTerm}_{\mathcal V}		
		\}
		\uplus
		\Facts{\LetTerm}_{\neg \mathcal V}
	\end{align*}
	
	\item In case $\vec v_1$ has an arrow variable $f\in \FactVar{\Facts{\LetTerm_1}_{\neg \mathcal V}}\setminus\FactVar{\Facts{\LetTerm_1}_{\mathcal V}}$, by the inductive hypothesis we get that $\Facts{\LetG{\vec v_1=e_1}{\VA(\LetTerm_1, \mathcal V)}}$ is equal to: 	
	\begin{align*}
	&
	\left\{
		\FSum{f}{
			\Fact{\vec v_1=e_1}
			\FProd
			\BigFProd(\Facts{\LetTerm_1}_{\neg \mathcal V})_f		
		}
	\right\}
	\uplus
	\left\{
		\BigFProd\Facts{\LetTerm_1}_{\mathcal V}
	\right\}
	\uplus
	(\Facts{\LetTerm_1}_{\neg\mathcal V})_{\neg f}
	\\
	&=
	\left\{
		\BigFProd\Facts{\LetTerm}_{\mathcal V}
	\right	\}
	\uplus
	\left\{
			\FSum{f}{
			\Fact{\vec v_1=e_1}
			\FProd
			\BigFProd\Facts{\LetTerm_1}_f		
			}
	\right	\}
		\uplus
		\Facts{\LetTerm}_{\neg (\mathcal V\cup\{f\})}\\
	&=
		\left\{
			\BigFProd\Facts{\LetTerm}_{\mathcal V}		
		\right	\}
		\uplus
		\Facts{\LetTerm}_{\neg \mathcal V}
	\end{align*}

	\end{itemize}
\item Case 3 of Def.~\ref{definition:VA}: observe that $\VA(\LetTerm_1, \mathcal V\cap\FV{\LetTerm_1})$ is of the shape $\LetG{\vec v'=e'}{\LetTerm'}$ and notice that:
\begin{equation}\label{eq:soundness_VA_eq2b}
	\Fact{
			\Seq{\vec v_1,\vec v'}
			=
			\LetG{
				\vec v_1=e_1
			}{
				\Seq{\vec v_1, e'}
			}
	}
	=
	\Fact{\vec v_1 = e_1}
	\FProd
	\Fact{\vec v' = e'}
\end{equation}
Let us suppose that $(\vec v')^a=g$ for an arrow variable $g$ which is not an output: the case $\vec v'$ positive or $g$ being an output are easier variants. Let us write $\mathcal V_1=\mathcal V\cap\FV{\LetTerm_1}$. We have that $\Facts{\VA(\LetTerm, \mathcal V)}$ is equal to: 
\begin{align}
	\label{eq:soundness_VA_case2b_line2}
	&
	\{
		\FSum{g}{
				\Fact{\vec v = e_1}
				\FProd
				\Fact{\vec v' = e'}
				\FProd
				\BigFProd
				\Facts{\LetTerm'}_{g}
		}
	\}
	\uplus
	\Facts{\LetTerm'}_{\neg g}
	\\
	\label{eq:soundness_VA_case2b_line3}
	&=
	\{
		\Fact{\vec v = e_1}
		\FProd
		\FSum{g}{
				\Fact{\vec v' = e'}
				\FProd
				\BigFProd
				\Facts{\LetTerm'}_{g}
		}
	\}
	\uplus
	\Facts{\LetTerm'}_{\neg g}
	\\
	\label{eq:soundness_VA_case2b_line4}
	&=
	\{
		\Fact{\vec v = e_1}
		\FProd
		\VA(\LetTerm_1, \mathcal V_1)_{\mathcal V_1}
	\}
	\uplus
	\Facts{\VA(\LetTerm_1, \mathcal V_1)}_{\neg \mathcal V_1}
	\\
	\label{eq:soundness_VA_case2b_line5}
	&=
	\left\{
		\Fact{\vec v = e_1}
		\FProd
		\BigFProd\Facts{\LetTerm_1}_{\mathcal V_1}\right\}
		\uplus
		\Facts{\LetTerm_1}_{\neg \mathcal V_1}
	\\
	\label{eq:soundness_VA_case2b_line6}
	&=
	\left\{\BigFProd\Facts{\LetTerm}_{\mathcal V}\right\}
	\uplus
	\Facts{\LetTerm}_{\neg \mathcal V}
\end{align}

Line~\eqref{eq:soundness_VA_case2b_line2} uses Equation~\eqref{eq:soundness_VA_eq2b}.
Line~\eqref{eq:soundness_VA_case2b_line3} applies the properties stated in Proposition~\ref{prop:factor_prod_ass_distr}.
Line~\eqref{eq:soundness_VA_case2b_line4} uses the hypothesis that $\mathcal V\subseteq\Fam(\Fact{\vec v' = e'})\setminus\Fam(\Facts{\LetTerm'})$. 
Finally, Line~\eqref{eq:soundness_VA_case2b_line5}  applies the inductive hypothesis and Line~\eqref{eq:soundness_VA_case2b_line6}  Definition~\ref{def:factor_of_let_term}.

\item Case 4 of Def.~\ref{definition:VA}:  let denote $\VA(\LetTerm_1, (\mathcal V\cap\FV{\LetTerm_1})\cup\{f\})$ by $\LetG{\vec v'=e'}{\LetTerm'}$ and notice that $f\in\FV{e'}\setminus\FV{\LetTerm'}$ and:
\begin{equation*}
	\Fact{
			\Seq{\vec v_1^+,\vec v'}
			=
			\LetG{
				\vec v_1=e_1
			}{
				\Seq{\vec v_1^+, e'}
			}
	}
	=
	\FSum{f}{
		\Fact{\vec v_1 = e_1}
		\FProd
		\Fact{\vec v' = e'}
	}
\end{equation*}
The reasoning is then similar to the previous Case 3, by adding only a commutation between $\sum_f$ and $\sum_g$.
\end{itemize}
\Myqed
\end{proof}

\RED{
	\begin{proof}[Proof of Lemma~\ref{lemma:SD_soundness}]
		The three cases of Def.~\ref{definition:SwapDef} correspond respectively to the definitions of the three commutative rules $\Swap_1, \Swap_2, \Swap_3$. so that we have $\LetTerm\Red[\Swap]\SD(\LetTerm)$. The fact that $\SD(\LetTerm)$ is well-typed then follows by Prop.~\ref{prop:let_rewriting_subject}, and the equality $\Facts{\LetTerm}=\Facts{\SD(\LetTerm)}$ is a consequence of Lemma~\ref{lemma:facts-invariant}.\Myqed
	\end{proof}
}

\RED{\begin{lemma}[Rewriting into $\VA$]\label{lemma:VA_rewriting}
		Given a positive let-term $\LetTerm$ with $n\geq 1$ definitions and a subset $\mathcal V\subseteq\FV{\LetTerm}$ disjoint from the output variables of $\LetTerm$, let us denote by $\LetG{\vec v'=e'}{\LetTerm'}$ the let-term $\VA(\LetTerm, \mathcal V)$. We have that:
		\begin{enumerate}
		\item $\mathcal V\subseteq\FV{e'}\setminus\FV{\LetTerm'}$;
		\item $\LetTerm$ rewrites into $\VA(\LetTerm, \mathcal V)$ by applying at most $n$ steps among $\{\Swap_1,\Swap_2,\Swap_3,\Mult\}$ rewriting rules in Fig.~\ref{fig:let-rewriting};
		\item hence $\VA(\LetTerm, \mathcal V)$ is a well-typed term with  same type and free variables as $\LetTerm$.
		\end{enumerate}
		\end{lemma}
		\begin{proof}
		Item 1 and 2 are proven by induction on $\LetTerm$. Item 1 is a simple inspection of the cases of Def.~\ref{definition:VA}. Item 2 is obtained by remarking that, by Lemma~\ref{lemma:SD_soundness}, Case 2 adds one $\Swap$ step to reduction obtained by the inductive hypothesis; Case 3 adds one $\Mult$ step and Case 4 adds one $\Swap_3$ step. For this latter case, one should remark that the side condition $f\in\FV{e'}$ is met because, by item 1, $f\in\FV{e'}\setminus\FV{\LetTerm'}$. 
		
		Item 3 is then a consequence of item 2 and Prop.~\ref{prop:let_rewriting_subject}.\Myqed
		\end{proof}

		\begin{lemma}[$\VA$ soundness]\label{lemma:VA_semantic}
		Given $\LetTerm$ and $\mathcal V$ as in Def.~\ref{definition:VA}, we have:
		\[
			\Facts{\VA(\LetTerm, \mathcal V)} = 
			\begin{cases}
				\left\{\BigFProd\Facts{\LetTerm}_{\mathcal V}\right\}\uplus\Facts{\LetTerm}_{\neg \mathcal V}
				&\text{if $\mathcal V\neq\emptyset$},\\
				\Facts{\LetTerm}
				&\text{otherwise.}
				\end{cases}
		\]
		\end{lemma}
		\begin{proof}
		If $\mathcal V=\emptyset$, then $\VA(\LetTerm, \mathcal V)=\LetTerm$. Otherwise, $\Facts{\VA(\LetTerm, \mathcal V)} = 	\left\{\BigFProd\Facts{\LetTerm}_{\mathcal V}\right\}\uplus\Facts{\LetTerm}_{\neg \mathcal V}$ is proven by induction on $\LetTerm$, splitting according to the cases of Def.~\ref{definition:VA}. \Myqed
		\end{proof}
		}

	\RED{\begin{proof}[Proof of Prop.~\ref{prop:rewriting_VE_well-typed}]
			By induction on $\LetTerm$, inspecting the cases of Def.~\ref{def:VE}.
			Case 1 consists in remarking that $\LetTerm\Red[\Elimin_x]\VEL(\LetTerm, x)$. Case 2 applies Lemma~\ref{lemma:VA_rewriting} for obtaining a rewriting sequence to $\LetG{\vec v_1=e_1; \vec v'=e'}{\LetTerm'}$ and then it adds a $\Mult$ (Case 2(a)) or $\Swap_3$ (Case 2(b)) step and the final $\Elimin_x$ step. Case 3 uses the inductive hypothesis and then Lemma~\ref{lemma:SD_soundness}.\Myqed
	\end{proof}}

\subsection*{Proof of Theorem ~\ref{th:soundnessVE_Let}.}
\RED{\begin{proof}[Proof of Theorem~\ref{th:soundnessVE_Let}]
		The proof is by induction on $\LetTerm$ and splits depending on the cases of Def.~\ref{def:VE}. By taking the notation of that definition, we consider only the base case 2(b) and the induction case 3, the case 1 being immediate and the case 2(a) being  an easier variant of 2(b).

		\paragraph{Case 2(b)  of Def.~\ref{def:VE}.} If $\LetTerm$ is $\LetG{\vec v_1=e_1}{\LetTerm_1}$, with $\vec v_1^a= f$. Let $\vec v_1^+=\vec y$ and let $\VA(\LetTerm_1, \{x, f\})$ be denoted by $\LetG{\vec v'=e'}{\LetTerm'}$. Lemma~\ref{lemma:VA_semantic} gives:
		\[
		\Facts{\LetG{\vec v'=e'}{\LetTerm'}}
		=\left\{\BigFProd\Facts{\LetTerm_1}_{\{x, f\}}\right\}\uplus\Facts{\LetTerm_1}_{\neg \{x, f\}}.
		\]
		Suppose also that $(\vec v')^a= g$ for an arrow variable $g$ (which is not an output by hypothesis), the case $\vec v'$ positive being an easier variant.  By Lemma~\ref{lemma:VA_rewriting} (item 1) $\{x,f\}\cap\FV{\LetTerm'}=\emptyset$, so by Lemma~\ref{lemma:splitting_facts}\todom{R2: reference to appendix making the proof slightly non self contained here.}:
		\begin{align*}
			\FSum{g}{\Fact{\vec v'=e'}\FProd\Facts{\LetTerm'}_{g}}&=\BigFProd\Facts{\LetTerm_1}_{\{x, f\}},\\
			\Facts{\LetTerm'}_{\neg g}&=\Facts{\LetTerm_1}_{\neg \{x, f\}}. 	
		\end{align*}
		
		Notice that
		$
		\Fact{\Seq{\vec y,\vec v'}=\LetG{\vec v_1=e_1}{\Seq{\vec y, e'}}}=
		\FSum{x,f}{\Fact{\vec v_1=e_1}\FProd\Fact{\vec v'=e'}}
		$, so that:
		\begin{align}
			&\Facts{\VEL(\LetTerm, x)}
			\\
			\label{eq:ve_sound1}
			&=\Facts{
				\LetG{
					\Seq{\vec y,\vec v'}=\LetG{\vec v_1=e_1}{\Seq{\vec y, e'}}
				}{\LetTerm'}
			}
			\\
			\label{eq:ve_sound3}
			&=
			\{\FSum{g}{
				\FSum{x,f}{\Fact{\vec v_1=e_1}\FProd\Fact{\vec v'=e'}}
				\FProd
				\Fact{\LetTerm'}_{g}
			}\}
			\uplus
			\Facts{\LetTerm'}_{\neg g}	
			\\
			\label{eq:ve_sound4}
			&=
			\{\FSum{x,f}{
				\Fact{\vec v_1=e_1}
				\FProd
				\FSum{g}{
					\Fact{\vec v'=e'}
					\FProd
					\Fact{\LetTerm'}_{g}
				}
			}\}
			\uplus
			\Facts{\LetTerm'}_{\neg g}	
			\\
			\label{eq:ve_sound5}
			&=
			\FSum{x,f}{
				\Fact{\vec v_1=e_1}
				\FProd
				\BigFProd\Facts{\LetTerm_1}_{\{x, f\}}
			}
			\uplus
			\Facts{\LetTerm_1}_{\neg \{x, f\}}
			\\
			\label{eq:ve_sound6}
			&=
			\{\FSum{x}{
				\FSum{f}{
					\Fact{\vec v_1=e_1}
					\FProd
					\BigFProd\Facts{\LetTerm_1}_{f}
				}
				\FProd
				(\BigFProd\Facts{\LetTerm_1}_{\neg f})_{x}
			}\}
			\uplus
			\Facts{\LetTerm_1}_{\neg x}
			\\
			\label{eq:ve_sound7}
			&=
			\{\FSum{x}{
				\FSum{f}{
					\Fact{\vec v_1=e_1}
					\FProd
					\BigFProd\Facts{\LetTerm_1}_{f}
				}
				\FProd
				(\BigFProd\Facts{\LetTerm}_{\neg f})_{x}
			}\}
			\uplus
			\Facts{\LetTerm}_{\neg x}
			\\
			\label{eq:ve_sound8}
			&=
			\{\FSum{x}{
				\BigFProd\Facts{\LetTerm}_{x}
			}\}
			\uplus
			\Facts{\LetTerm}_{\neg x}
			\\
			\label{eq:ve_sound9}
			&=
			\VEF(\Facts{\LetTerm}, x)
		\end{align}
		Line~\eqref{eq:ve_sound3} uses the hypothesis that the arrow variable $g$ is in $\FactVar{\FSum{x,f}{\Fact{\vec v_1=e_1}\FProd\Fact{\vec v'=e'}}}$. 
		Line~\eqref{eq:ve_sound4} uses the properties given in Prop.~\ref{prop:factor_prod_ass_distr}.
		Line~\eqref{eq:ve_sound5} applies the equalities achieved above by using Lemma~\ref{lemma:splitting_facts}.
		Line~\eqref{eq:ve_sound6} uses again Prop.~\ref{prop:factor_prod_ass_distr} to split the summing out of $x,f$, to decompose 			$\BigFProd\Facts{\LetTerm_1}_{x,f}$ into the factor containing $f$ and the ones which do not, and to restrict the summing out of $f$ to the factor having this variable. 
		Line~\eqref{eq:ve_sound7}  replaces $\Facts{\LetTerm_1}_{\neg f}$ and $\Facts{\LetTerm_1}_{\neg x}$ with, respectively, $\Facts{\LetTerm}_{\neg f}$ and $\Facts{\LetTerm}_{}\neg x$, as they are the same sets (recall $f, x$ do appear in $\vec v_1$). 
		Line~\eqref{eq:ve_sound8} and Line~\eqref{eq:ve_sound9} follow easily from the definitions.

		\paragraph{Case 3 of Def.~\ref{def:VE}.} Let us consider the inductive case. Let denote $\LetTerm$ by $\LetG{\vec v_1=e_1}{\LetTerm_1}$ and suppose that $\VEL(\LetTerm, x)=\SD(\LetG{\vec v_1=e_1}{\VEL(\LetTerm_1, x))}$. We then have:
		\begin{align}
			\label{eq:ve_soundInd1}
			\Facts{\VEL(\LetTerm, x)}&=\Facts{\SD(\LetG{\vec v_1=e_1}{\VEL(\LetTerm_1, x))}}
			\\
			\label{eq:ve_soundInd2}
			&=\Facts{\Let{\vec v_1}{e_1}{\VEL(\LetTerm_1, x)}}
			\\
			\label{eq:ve_soundInd3}
			&=\FComp{\Fact{\vec v_1=e_1}}{\Facts{\VEL(\LetTerm_1, x)}}
			\\
			\label{eq:ve_soundInd4}
			&=\FComp{\Fact{\vec v_1=e_1}}{\VEF(\Facts{\LetTerm_1}, x)}
			\\
			\label{eq:ve_soundInd5}
			&=\VEF(\FComp{\Fact{\vec v_1=e_1}}{\Facts{\LetTerm_1}}, x)
			\\
			\label{eq:ve_soundInd6}
			&=\VEF(\Facts{\LetTerm},x)
		\end{align}
		where\todom{R2: the "::" notation doesn't seem to work as an abbreviation of the different cases of definition 3.10.  I think it would be better to spell out the separate cases, or clarify in what sense the reasoning with the abstract operation :: can be seen as subsuming the possible cases.} $\FComp$ is one of the cases of Def.~\ref{def:factor_of_let_term}, depending on whether the first definition contains an arrow variable or not.  Line~\eqref{eq:ve_soundInd2} uses Lemma~\ref{lemma:SD_soundness}, Line~\ref{eq:ve_soundInd3} applies the properties stated in Prop.~\ref{prop:factor_prod_ass_distr}; Line~\ref{eq:ve_soundInd4} follows from the inductive hypothesis and we build  $\VEF(\Facts{\LetTerm},x)$ by using again Prop.~\ref{prop:factor_prod_ass_distr}.\Myqed
	\end{proof}
}


\subsection*{Proofs for section~\ref{subsect:complexity}: Complexity Analysis}


\begin{lemma}\label{lemma:fam_vars}
	Given $\LetTerm\Def\LetG{\vec v_1=e_1}{\LetTerm_1}$, with $\vec x$ be the sequence (possibly empty) of the positive variables in $\vec v_1$, then, for any set of variables $\mathcal V$, we have that:
	\begin{equation*}
		\left(\FactVar{\Facts{\LetTerm_1}_{\mathcal V}}\setminus\FV{\LetTerm_1}\right)
		\uplus
		\left(\FactVar{\Facts{\LetTerm}_{\mathcal V}}\cap\FV{\vec v_1}\right)
		\subseteq
		\FactVar{\Facts{\LetTerm}_{\mathcal V}}\setminus\FV{\LetTerm}\,.
	\end{equation*}
\end{lemma}
\begin{proof}
	First, notice that $\left(\FactVar{\Facts{\LetTerm_1}_{\mathcal V}}\setminus\FV{\LetTerm_1}\right)$ and 
	$\left(\FactVar{\Facts{\LetTerm}_{\mathcal V}}\cap\FV{\vec v_1}\right)$ are disjoint, as, by renaming, we can always suppose that a bounded occurrences of a variable in $\LetTerm_1$ is distinct from any variable of $\vec v_1$. 
	
	Suppose $v\in \FactVar{\Facts{\LetTerm_1}_{\mathcal V}}\setminus\FV{\LetTerm_1}$, then we have $v\in \FactVar{\Facts{\LetTerm}_{\mathcal V}}$ as the only possible variable in $\FactVar{\Facts{\LetTerm_1}_{\mathcal V}}$ pulled out in $\FactVar{\Facts{\LetTerm}_{\mathcal V}}$ is an arrow variable $f$ in $\vec v_1$ that has a free occurrence in $\LetTerm_1$ (and so it is not in $\FactVar{\Facts{\LetTerm_1}_{\mathcal V}}\setminus\FV{\LetTerm_1}$). Moreover, suppose $v\in\FV{\LetTerm}$, since by hypothesis $v\notin\FV{\LetTerm_1}$, we deduce that $v\in\FV{e_1}$, \todom{hypothesis on renaming}which means that $v$ has both a free occurrence in $e_1$ as well as a bound occurrence in $\LetTerm_1$, a case that can be always avoided by renaming. We conclude $v\in \FactVar{\Facts{\LetTerm}_{\mathcal V}}\setminus\FV{\LetTerm}$.
	
	Suppose $v\in\FactVar{\Facts{\LetTerm}_{\mathcal V}}\cap\FV{\vec v_1}$, then clearly $v\in \FactVar{\Facts{\LetTerm}_{\mathcal V}}\setminus\FV{\LetTerm}$ as $\FV{\vec v_1}$ and $\FV{\LetTerm}$ are disjoint. \Myqed
\end{proof}

\begin{lemma}\label{lemma:size_SD}
Given a let-term $\LetTerm\Def\LetG{\vec v_1=e_1;\vec v_2=e_2}{\LetTerm'}$ with at least two definitions. Let $\vec x$ the set of positive variables in $\FV{\vec v_1}\cap\FV{e_2}$ or $\vec v_1^+$ in case $\vec v_1^a=f$ and $f\in\FV{e_2}$. We have that:
\begin{align*}
	\Size(\SD(\LetTerm))&\leq
		2 + 2\Size(\vec x) + \Size(\LetTerm)\,.
\end{align*}  
\end{lemma}
\begin{proof}
By inspecting the cases of Def.~\ref{definition:SwapDef}.\Myqed
\end{proof}

\begin{lemma}\label{lemma:size_VA}
Given a positive let-term $\LetTerm\Def\LetG{\vec v_1=e_1}{\LetTerm_1}$ with at least one definition and given a subset $\mathcal V$ of $\FV{\LetTerm}$ disjoint from the output of  $\LetTerm$. We have:
\[
	\Size(\VA(\LetTerm, \mathcal V))
	\leq
	\Size(\LetTerm) 
	+ 4 \times \Size(\FactVar{\Facts{\LetTerm}_{\mathcal V}}\setminus\FV{\LetTerm})\,.
\]
\end{lemma}
\begin{proof}
The proof is by induction on $\LetTerm\Def\LetG{\vec v_1=e_1}{\LetTerm_1}$, splitting into the cases of Definition~\ref{definition:VA}, from which we adopt the notation.
\begin{itemize}
\item Case 1. We have: $\Size(\VA(\LetTerm, \mathcal V))=\Size(\LetTerm)$. The above inequality is then immediate. 

\item Case 2, \ie $\mathcal V\cap\FV{e_1}=\emptyset$ and $\mathcal V\cap\FV{\LetTerm_1}\neq\emptyset$. Let $\VA(\LetTerm_1, \mathcal V)=\LetG{\vec v' = e'}{\LetTerm'}$ and let $\vec x$ be the positive variables in $\FV{\vec v_1}\cap\FV{e'}$. Notice that $\vec x\subseteq\FactVar{\Facts{\LetTerm}_{\mathcal V}\cap\FV{\vec v_1}}$, so Lemma~\ref{lemma:fam_vars} gives $\Size(\FactVar{\Facts{\LetTerm_1}_{\mathcal V}}\setminus\FV{\LetTerm_1})+\Size(\vec x)\leq\Size((\FactVar{\Facts{\LetTerm}_{\mathcal V}}\setminus\FV{\LetTerm}))$. By applying Lemma~\ref{lemma:size_SD}, we then have:
		\begin{align*}
		&\Size(\VA(\LetTerm, \mathcal V))\\
		&\leq 2 + 2\Size(\vec x) + \Size(\vec v_1) + \Size(e_1) + \Size(\VA(\LetTerm_1, \mathcal V))\\
		&\leq 2 + 2\Size(\vec x) + \Size(\vec v_1) + \Size(e_1) + \Size(\LetTerm_1) 
	+ 4 \times \Size(\FactVar{\Facts{\LetTerm_1}_{\mathcal V}}\setminus\FV{\LetTerm_1})\\
		&\leq \Size(\LetTerm) 
	+ 2 + 2\Size(\vec x) +  4 \times \Size(\FactVar{\Facts{\LetTerm_1}_{\mathcal V}}\setminus\FV{\LetTerm_1})\\
			&\leq \Size(\LetTerm) 
	+ 2 + 2\Size(\vec x) +  4 \times (\Size(\FactVar{\Facts{\LetTerm}_{\mathcal V}}\setminus\FV{\LetTerm}) - \Size(\vec x))\\
			&\leq \Size(\LetTerm) +  4 \times \Size(\FactVar{\Facts{\LetTerm}_{\mathcal V}}\setminus\FV{\LetTerm}).
		\end{align*}

	\item Case 3, \ie $\mathcal V\cap\FV{e_1}\neq \emptyset$ and $\vec v_1=\vec x$ is positive. Again we can apply Lemma~\ref{lemma:fam_vars} and getting  $\Size(\FactVar{\Facts{\LetTerm_1}_{\mathcal V}}\setminus\FV{\LetTerm_1})+\Size(\vec x)\leq\Size((\FactVar{\Facts{\LetTerm}_{\mathcal V}}\setminus\FV{\LetTerm}))$. So we have:
		\begin{align*}
		&\Size(\VA(\LetTerm, \mathcal V))\\
		&= \Size(\vec x) + \Size(\vec v') + \Size(e_1) + \Size(e') + \Size(\LetTerm')\\
		&= \Size(\vec x) +  \Size(e_1)  + \Size(\VA(\LetTerm_1, \mathcal V))\\
		&= \Size(\vec x) +  \Size(e_1)  + \Size(\LetTerm_1) + 4 \times \Size(\FactVar{\Facts{\LetTerm_1}_{\mathcal V}}\setminus\FV{\LetTerm_1})\\
		&\leq \Size(\LetTerm) + 4 \times \Size(\FactVar{\Facts{\LetTerm_1}_{\mathcal V}}\setminus\FV{\LetTerm_1})\\
		&\leq  \Size(\LetTerm) +  4 \times \Size(\FactVar{\Facts{\LetTerm}_{\mathcal V}}\setminus\FV{\LetTerm}).
		\end{align*}
		
	\item Case 4, \ie $\mathcal V\cap\FV{e_1}\neq \emptyset$ and $\vec v_1=\Seq{\vec x,f}$. This case is analogous to the previous one. 
\Myqed
\end{itemize}
\end{proof}

\begin{proof}[Proof of Proposition~\ref{prop:size_VE}]
By induction on  $\LetTerm$, splitting in the cases of  Def.~\ref{def:VE}, of which we adopt the notation. Case 1 is immediate. Case 2(a) is an easier version of Case 2(b) and we omit to detail it.
\begin{itemize}
\item Case 2(b). First, notice that,
 $\Size(\FactVar{\Facts{\LetTerm_1}_{\{x,f\}}}\setminus\FV{\LetTerm_1})+\Size(\vec y)\leq\Size(\FactVar{\Facts{\LetTerm}_x}\setminus\FV{\LetTerm})$. By using Lemma~\ref{lemma:size_VA} we can then argue:
\begin{align*}
	&\Size(\VEL(\LetTerm, x))\\
	&= 2\Size(\vec y) + \Size(\vec v') + \Size(\vec v_1) + \Size(e_1) + \Size(e') + \Size(\LetTerm')\\
	&\leq 2\Size(\vec y) + \Size(\vec v_1) + \Size(e_1) + \Size(\VA(\LetTerm_1,\{x,f\}))\\
	&\leq 2\Size(\vec y) + \Size(\vec v_1) + \Size(e_1) + \Size(\LetTerm_1) + 4 \Size(\FactVar{\Facts{\LetTerm_1}_{\{x,f\}}}\setminus\FV{\LetTerm_1})\\
	&\leq 2\Size(\vec y) + \Size(\LetTerm) + 4 \Size(\FactVar{\Facts{\LetTerm}_x}\setminus\FV{\LetTerm})-4\Size(\vec y)\\
	&\leq \Size(\LetTerm) + 4 \Size(\FactVar{\Facts{\LetTerm}_x}\setminus\FV{\LetTerm})
\end{align*}
\item Case 3. Let $\VEL(\LetTerm_1,x)\Def\LetG{\vec v'=e'}{\LetTerm'}$ and let $\vec y$ be the sequence of the positive variables in $\vec v_1\cap\FV{e'}$. Again, 
we notice that
$\Size(\FactVar{\Facts{\LetTerm_1}_{\{x\}}}\setminus\FV{\LetTerm_1})+\Size(\vec y)\leq\Size(\FactVar{\Facts{\LetTerm}_x}\setminus\FV{\LetTerm})$. 
By applying Lemma~\ref{lemma:size_SD},  we have:
\begin{align*}
	&\Size(\VEL(\LetTerm, x))\\
	&\leq 2 + 2\Size(\vec y) + \Size(\vec v_1) + \Size(e_1) + \Size(\VEL(\LetTerm_1,x))\\
	&\leq 2 + 2\Size(\vec y) + \Size(\vec v_1) + \Size(e_1) +  \Size(\LetTerm_1) + 4 \Size(\FactVar{\Facts{\LetTerm_1}_x}\setminus\FV{\LetTerm_1})\\
	&\leq 2 + 2\Size(\vec y) +  \Size(\LetTerm) + 4 \Size(\FactVar{\Facts{\LetTerm}_x}\setminus\FV{\LetTerm})-4\Size(\vec y)\\
	&\leq \Size(\LetTerm) + 4 \Size(\FactVar{\Facts{\LetTerm}_x}\setminus\FV{\LetTerm}).
\end{align*}
\end{itemize}
\Myqed
\end{proof}

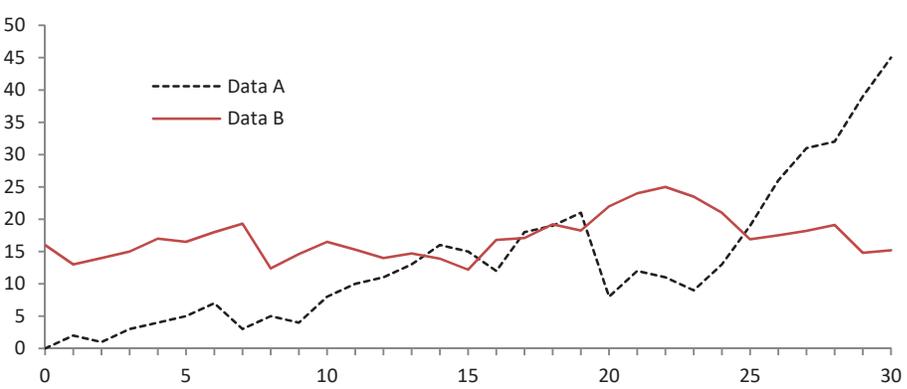
\end{document}